\documentclass[12pt]{article}
\usepackage{amsfonts}
\usepackage{amsmath}
\usepackage{epigraph}
\usepackage{amssymb} 
\usepackage{graphicx}
\usepackage{amsthm} 
\makeatletter
\def\th@plain{%
  \thm@notefont{}%
  \itshape %
}
\def\th@definition{%
  \thm@notefont{}%
  \normalfont %
} 
\makeatother
\usepackage{array}
\usepackage{indentfirst}
\usepackage[left=1in,top=1in,bottom=1in,right=1in,lines=32]{geometry}
\usepackage[round]{natbib}
\usepackage{graphicx}
\usepackage{verbatim}
\usepackage{enumerate}
\usepackage[export]{adjustbox}
\usepackage{accents}
\usepackage{float}
\usepackage{verbatim}
\usepackage[usenames,dvipsnames,svgnames]{xcolor}
\usepackage{graphics,color}
\usepackage{tikz}
\usetikzlibrary{trees, positioning, arrows.meta, fit}
\usepackage{mathpazo}
\usepackage{xr-hyper}
\usepackage[nodisplayskipstretch]{setspace}
\usepackage[%
plainpages,%
linkcolor=blue,
colorlinks=true,%
urlcolor=blue,%
filecolor=black,%
citecolor=black,%
pdfpagemode=UseOutlines,%
pdfauthor={Cole Wittbrodt},%
pdfsubject={Infinite Dimensional Delegation},%
]{hyperref}
\usepackage{authblk}
\newtheoremstyle{pfof}
{\topsep}
{\topsep}
{\normalfont}
{0pt}
{\bfseries}
{}
{5pt plus 1pt minus 1pt}
{}
\theoremstyle{pfof}
\newtheorem*{proofof}{Proof of}
\theoremstyle{definition}
\newtheorem*{defn}{Definition}

\newtheorem{claim}{Claim}

\newtheorem{ex}{Example}
\theoremstyle{plain}
\newtheorem{thm}{Theorem}
\newtheorem{prop}{Proposition}
\newtheorem{cor}{Corollary}
\newtheorem{lem}{Lemma}
\newtheorem{asn}{Assumption}

\theoremstyle{remark}

\DeclareMathOperator{\supp}{supp}

\title{Bargaining with Absentmindedness}
\author{Cole Wittbrodt\thanks{Wittbrodt: Columbia University, Graduate School of Arts and Sciences, Department of Economics; E-mail: \texttt{cole.wittbrodt@columbia.edu}. This paper was previously circulated under the title \textit{What Slips the Mind Stalls the Deal: Delay in Bargaining with Absentmindedness}. I thank Antoine Chapel, Yeon-Koo Che, Grace Chuan, Laura Doval, Michael Lee, Qingmin Liu, and Alessandro Pavan for helpful comments and feedback.}} 

\begin{document}
\maketitle

\begin{abstract}
   Delay is the norm in bargaining. I propose a novel source of bargaining delay: absentmindedness. Instead of interpreting absentmindedness as a literal memory friction, I use absentmindedness to represent a broader form of bounded rationality in dynamic games where players cannot perfectly track a game's progression. I show that equilibrium behavior is equivalent to that of a bargainer who commits ex-ante to a stationary behavioral strategy. Bargainers unable to finely condition play on the stage of dynamic interaction can credibly refuse last-minute ultimatums. Other parties that anticipate this behavior are driven to offer preemptive concessions to avoid a breakdown in negotiations. Absentmindedness is thus a source of bargaining power, even for players who never make offers. This bargaining power comes at the cost of efficiency; I show that there can be equilibria where offers are rejected on the path of play.
\end{abstract}

\newpage

\vspace{0.5cm} %
\noindent %
\begin{quote}
\itshape Jarndyce and Jarndyce drones on. This scarecrow of a suit has, over the course of time, become so complicated, that no man alive knows what it means. The parties to it understand it least... Scores of persons have deliriously found themselves made parties in Jarndyce and Jarndyce without knowing how or why.

\hfill --- {\normalfont \text{Charles Dickens,} \textit{Bleak House}}
\end{quote}
\vspace{0.5cm} %

\section{Introduction}

Many classic finite-horizon complete-information bargaining models fail to generate delay, despite the fact that delay often occurs in reality. %
Backward induction is simply too powerful. Players know how the game will end, and prefer to hasten the inevitable. To reconcile the theory with the empirical reality of delayed agreements, one might question whether agents truly reason backwards from the end of the game. This concern is especially relevant when the time horizon is long. Experimental evidence from sequential bargaining is consistent with this concern.\footnote{ \cite{johnsoncamerersenrymon2002} directly measure subjects' information search in a three-round bargaining game and find patterns of limited lookahead rather than backward induction. Subjects who are explicitly trained to use backward induction behave substantially more in line with its predictions.} Hence, the literature has often turned to infinite-horizon models to capture such situations, but these models seldom generate delay without also introducing incomplete information.

This paper takes a different perspective. I introduce a finite-horizon model of bargaining where players cannot perfectly track the progress of the game. In particular, I identify \textit{absentmindedness} as a potential source of bargaining delay. Absentmindedness is a form of imperfect recall; in a dynamic setting absentminded agents do not remember their own past actions, or they forget information that they once knew. Rather than interpreting this as a literal lapse in memory, I view absentmindedness as a conceptual tool that captures this broader phenomenon of bounded rationality in dynamic games. Specifically, it reflects the idea that agents may be unable to finely condition on the exact stage of a dynamic interaction.\footnote{Markov strategies also do not condition on the precise stage of a dynamic interaction, but are typically considered in an environment where agents \textit{could} condition on time or histories. In my environment, players cannot condition in this way. This amounts to a restriction on the space of possible deviations in games, and affects how players evaluate the payoffs from such deviations. Predictions differ substantially in both frameworks.} Indeed, I show that equilibrium behavior is equivalent to that of a bargainer who chooses a stationary behavioral strategy ex-ante and applies it throughout the negotiation.

This interpretation need not be purely cognitive. For instance, consider a bargaining committee that enters a negotiation with a mandate specifying which proposals it may accept. The committee may know the calendar date perfectly well, but its acceptance rule may nevertheless be insensitive to the precise number of bargaining opportunities that remain. If the opposing party understands this constraint, the mandate can make last-minute ultimatum threats less effective. Absentmindedness provides a simple way to represent this form of coarse contingent planning.

The mechanism by which absentmindedness can result in inefficient delay is straightforward. Bargainers who do not know whether the final period, or trade deadline, has arrived can \textit{credibly reject last minute offers}, rendering them immune to standard ultimatum-bargaining threats. Crucially, this bounded rationality is understood by the other party. Since an absentminded party can reject an offer at the "eleventh hour", the other party may make concessions in equilibrium. This results in an equilibrium feedback: the possibility of concessions by the other party encourages absentminded players to reject unfavorable offers. Such offers can be made on the equilibrium path, hence there is a possibility for inefficient delay.

I demonstrate this mechanism by first considering the simplest possible model of negotiation. I consider a two-player finite-horizon bargaining game where one player (the proposer) offers a course of action from a binary set $\{a_P,a_R\}$, and the other (the respondent) decides whether to accept or reject the offer in each period. The proposer and respondent disagree over their preferred course of action; the proposer prefers to take action $a_P$ and the respondent prefers to take action $a_R$. If the game ends before an offer is accepted, the deal falls through and both parties earn a payoff of $0$.

With perfect recall, backward induction yields a unique solution to this game: immediate agreement. If, for instance, the respondent is absentminded, backward induction fails; the respondent does not necessarily accept an offer in the final period since he fails to realize that the trade deadline looms. To analyze this environment, I characterize all equilibria when parties are sufficiently patient. As in the perfect recall case, there is an equilibrium where the proposer always makes her preferred offer, and the respondent always accepts. There is no equilibrium delay. Unlike the perfect recall case, there is also an equilibrium where the proposer always offers $a_R$, and the respondent would always reject $a_P$. Since this event is off-path, if the respondent believes that he is not at the trade deadline (with sufficiently high probability) upon receiving an offer of $a_P$, he will rationally reject. Hence, the respondent has bargaining power in this equilibrium. Much like the other pure strategy equilibrium, there is no delay in the respondent-preferred equilibrium.

The last and arguably most conceptually relevant equilibrium involves both parties randomizing. The proposer mixes between $a_P$ and $a_R$ only in the final period of the game, as the trade deadline looms. In every other period, the proposer offers $a_P$. The respondent randomizes between rejecting and accepting $a_P$, balancing the value of accepting the offer of $a_P$ immediately against the value of potentially receiving his preferred offer at a later date. The proposer attempts repeatedly to convince the respondent to concede by offering $a_P$ until the final period. At this point, the proposer considers offering $a_R$ to prevent the deal from falling through. When the proposer offers $a_R$, she concedes and reveals her private information that the deadline has arrived. There is, of course, positive probability of delay in this equilibrium.

Both the respondent-preferred equilibrium and the equilibrium with randomization do not have natural analogues in the perfect recall case. The respondent's \textit{rational confusion} is the source of his bargaining power; since the respondent can rationally reject unfavorable offers in the final period of the game, the proposer considers sending a favorable offer to prevent the deal from falling through. Moreover, deals are made most frequently immediately prior to the deadline.

The probability of delay in the mixing equilibrium does not depend on the discount factor nor on the time horizon (provided players are sufficiently patient).\footnote{With imperfect recall, the classic equivalence between behavioral strategies and mixed strategies \citep{kuhn} breaks down. Here, I refer to the "mixing equilibrium" as the equilibrium where the respondent plays a behavioral strategy which randomizes over accepting and rejecting $a_P$.}  The probability of reaching a deal immediately is equal to the ratio of the proposer's payoff when action $a_R$ is chosen to her payoff when action $a_P$ is chosen in order to hold the proposer indifferent between offering $a_R$ and $a_P$ in the final period. The probability the deal falls through, however, is increasing in the discount factor. When the respondent is more patient, the cost of rejecting an offer of $a_P$ falls. To keep the respondent indifferent between accepting and rejecting, the proposer must concede less frequently in the final period of the game, thus increasing the chance that the deal falls through.

I also study two extensions. First, I consider alternative bargaining protocols. Delay continues to arise when the proposer is absentminded, though only for intermediate discount factors; when both parties are absentminded, the results are qualitatively similar. Second, I allow parties to bargain over how to split the surplus using side payments. The ability to fine-tune offers weakens the mechanism: with impatient players, Markov equilibria feature immediate agreement, while equilibria with delay reemerge when players are patient and strategies may depend on history. Section \ref{S:extensions} briefly discusses these results; formal statements and proofs are in the Supplemental Material.

\subsection*{Related Literature}

This paper contributes to an extensive literature on bargaining, emanating from the classic setting of \cite{rubinstein82}, which shows that a deal is reached immediately under complete information. In one-sided incomplete information settings, where the uninformed party makes offers (\cite{gsw} and \cite{flt}), delay vanishes in the patient limit as deals are reached in the "twinkle of an eye" \citep{coase} in any weakly stationary equilibrium. Bargaining models with two-sided asymmetric information can generate delay when gains from trade are not common knowledge (\cite{crampton1984}, \cite{cho1990}, \cite{crampton1991}). Another approach to generating delay with information asymmetry is \textit{higher-order uncertainty}, where players are uncertain about the other's beliefs over their valuation \citep{feinbergandy}.

The former papers consider information asymmetries regarding the underlying value of trade. The present article takes a different approach: there is common knowledge of the value of a deal, but uncertainty regarding the trade environment itself. In this sense, the present article is most similar to the literature on \textit{reputational bargaining}, where there are obstinate bargainer types who only accept certain offers. \cite{abreugul2000} study the reputational bargaining framework in an infinite horizon setting, and \cite{fanningjmp} extends \cite{abreugul2000} to a finite-horizon setting to study the deadline effects.\footnote{\cite{spier1992} shows a similar deadline effect, but under a different economic mechanism where one party \textit{prefers} to delay the deal.} The economic mechanism highlighted in the present article is similar to that in the reputational bargaining literature: players may delay in hope that the other party makes a concession. While not a reputational bargaining paper, \cite{friedenberg2019} also shows bargaining delay can arise without invoking private values by demonstrating the role of \textit{strategic uncertainty} in bargaining games. Finally, a working paper version of \cite{ravidaer} studies rational inattention in buyer-seller bargaining, and shows that equilibrium delay can arise. The mechanism is similar in that the bounded rationality gives the buyer bargaining power, at the cost of efficiency.

This paper also contributes to the literature on games of imperfect recall. To my knowledge, few papers study absentmindedness in economic games. Most existing work studies a particular decision problem, the so-called \textit{absentminded driver problem} (\cite{pr}, \cite{ahp}, and \cite{gilboa1997}). A notable exception is \cite{lambertmarpleshoham2019}, which extends classic solution concepts, namely agent equilibria \citep{strotz55}, sequential equilibria \citep{krepswilson82}, and perfect equilibria \citep{selten75} to games with imperfect recall. My definition of equilibria is a special case of the version of agent equilibria defined in \cite{lambertmarpleshoham2019} --- though I also show that equilibria satisfy a sequential equilibrium style refinement. \cite{hillaskvasov} also define solution concepts in games with imperfect recall, but my solution concept is closer to those in \cite{lambertmarpleshoham2019}. Two recent applications of the solution concepts in \cite{lambertmarpleshoham2019} are \cite{ChenManuscript-CHEIRA-6} and \cite{szenteslevy}, which study imperfect recall games in the context of artificial intelligence.

\section{Model}\label{S:baseline}

Two parties negotiate over a course of action. There is a binary set of actions $\{a_P,a_R\}$. One player is the proposer $(P)$, and offers a course of action in each period $t=1,...,T$. If accepted by the other player (the respondent $R$), the proposed action is implemented. If a deal is reached, undiscounted payoffs are given by Table \ref{tab:placeholder}:

\begin{table}[h!]
    \centering
    \large
    \begin{tabular}{c|cc}
         & $a_P$ & $a_R$\\\hline
       P  & $u_P^P$ & $u_R^P$\\
        R & $u_P^R$ & $u_R^R$ \\
    \end{tabular}
    \caption{Payoff Matrix. Rows: Players. Columns: Actions.}
    \label{tab:placeholder}
\end{table}
Subscripts on payoffs always refer to the action and superscripts refer to the player. Assume, for now, that all payoffs are strictly positive (that is, $u_i^P,u_i^R>0$ for $i=P,R$). Each player discounts future payoffs with a discount factor $0<\delta\leq 1$. That is, if action $a_i$ is accepted in period $t$, payoffs are $\delta^{t-1}u_i^P$ and $\delta^{t-1}u_i^R$ for the proposer and respondent respectively. If no offer is accepted by the end of period $T$, both parties receive a payoff of $0$.

The respondent is \textit{absentminded}.\footnote{See, e.g., \cite{pr}.} She cannot \textit{a priori} distinguish between the time periods $t=1,...,T$, and cannot condition her strategy on the history of the game. The proposer has perfect recall. This environment most clearly demonstrates the relationship between imperfect recall and bargaining, since the respondent has no bargaining power when there is perfect recall. A number of extensions are studied later: namely the case where the proposer has imperfect recall (Appendix \ref{S:propforg}) and the case where both players have imperfect recall (Appendix \ref{S:bothforg}). Figure \ref{fig:game_tree} displays the game tree when $T=2$. 

\begin{figure}[h!]
    \centering
    \begin{tikzpicture}[
        grow=right,
        level distance=2.5cm,
        sibling distance=1.5cm,
        every node/.style = {font=\small},
        edge from parent/.style = {draw},
        edge from parent path={(\tikzparentnode.east) -- (\tikzchildnode.west)},
        dot/.style = {circle, fill, inner sep=1.5pt}
    ]

    \node[dot, label=left:{$t=1$}] {}
        child {node[dot] (bottomnodeone) {} edge from parent node[below] {$a_R$}}
        child {node[dot] (topnodeone) {} edge from parent node[above] {$a_P$}};

    \draw (topnodeone) -- ++(0,1) node[above] {$(u_P^P, u_P^R)$} node[midway, left] {$A_P$};
    \draw (bottomnodeone) -- ++(0,-1) node[below] {$(u_R^P, u_R^R)$} node[midway, left] {$A_R$};

    \draw (topnodeone) -- ++(2,0) coordinate (topcont) node[midway, above] {$R_P$};
    \draw (bottomnodeone) -- ++(2,0) coordinate (bottomcont) node[midway, below] {$R_R$};

    \draw (topcont) -- + (2,0) coordinate (topend) node[midway, above] {$a_P$};
    \draw (bottomcont) -- + (2,0) coordinate (bottomend) node[midway, below] {$a_R$};

    \node[dot, label=above:{$t=2$}] at (topcont) {};
    \node[dot, label=below:{$t=2$}] at (bottomcont) {};
    \node[dot] at (topend) {};
    \node[dot] at (bottomend) {};

    \draw (topcont) -- (bottomend) node[midway, right, xshift=5pt] {$a_R$};
    \draw (bottomcont) -- (topend) node[midway, left, xshift=-5pt] {$a_P$};

    \draw (topend) -- ++(0,1) node[above] {$(\delta u_P^P, \delta u_P^R)$} node[midway, left, yshift=4pt] {$A_P$};
    \draw (bottomend) -- ++(0,-1) node[below] {$(\delta u_R^P,\delta u_R^R)$} node[midway, left, yshift=-4pt] {$A_R$};
    \draw (topend) -- ++(2,0) node[midway, above] {$R_P$} node[right] {$(0,0)$};
    \draw (bottomend) -- ++(2,0) node[midway, above] {$R_R$} node[right] {$(0,0)$};

    \draw[dashed, bend left=45] (topnodeone) to (topend);
    \draw[dashed, bend right=45] (bottomnodeone) to (bottomend);

    \end{tikzpicture}
    \caption{\small Game tree when $T=2$. The nodes after receiving an offer $a_P$ are in the same information set. The nodes after receiving an offer $a_R$ are in the same information set. In each payoff pair, the first element represents the proposer's payoffs and the second element represents respondent payoffs.}
    \label{fig:game_tree}
\end{figure}

A proposer strategy is a $T$-tuple $(\sigma_1,...,\sigma_T)$ where $\sigma_t$ is the probability of offering $a_P$ at time $t$. Implicitly, I assume that the proposer does not condition her action in period $t$ on their previous offers. That is, a proposer strategy is a Markov strategy.\footnote{By Markov strategy, I mean that the strategy does not depend on the payoff irrelevant history, but can depend on calendar time. Calendar time is payoff relevant since the game is finite horizon.} This is without loss of generality in the binary action case, as I'll argue below in Section \ref{S:analysis}.

Suppose throughout that the respondent strictly prefers $a_R$ to $a_P$, that is $u_R^R>u_P^R$. Since it is weakly dominant for the respondent to accept an offer to play his preferred action, assume throughout that he accepts an offer to play $a_R$. Thus, a respondent strategy is a value $p\in[0,1]$ representing the probability of accepting an offer $a_P$. Observe that I define strategies as behavioral strategies, rather than mixed strategies. These notions are not equivalent in games of imperfect recall \citep{isbell1957}. There are numerous issues when considering mixed strategies in settings with imperfect recall \citep{pr}, so I focus on behavioral strategies for the purpose of this analysis. 

I assume throughout that the proposer and the respondent are \textit{misaligned}. That is, they each prefer a different course of action. Hence, the proposer strictly prefers $a_P$ to $a_R$. If the proposer and respondent had aligned preferences, the only reasonable outcome is instant agreement on the mutually preferred outcome.

\begin{asn}[Misalignment]
    The proposer and respondent are misaligned if $u_P^P>u_R^P$ and $u_R^R>u_P^R$.
\end{asn}

Let $\alpha_t^P$ denote the respondent's belief that he is in period $t$, given he receives the offer $a_P$. Similarly, let $\alpha_t^R$ denote the respondent's belief that he is in period $t$, given he receives the offer $a_R$. The $T$-tuples of conditional beliefs over calendar time $\alpha^P=(\alpha_1^P,...,\alpha_T^P)$ and  $\alpha^R=(\alpha_1^R,...,\alpha_T^R)$ are consistent with $(\sigma_1,...,\sigma_T)$ and $p$ if
\begin{align*}
  \begin{aligned}
    \alpha_t^P &= \frac{\sigma_t\gamma_{t}}{\sum_{k=1}^T\sigma_k\gamma_k}
  \end{aligned}
  \qquad
  \begin{aligned}
    \alpha_t^R &= \frac{(1-\sigma_t)\gamma_t}{\sum_{k=1}^T(1-\sigma_k)\gamma_k}
  \end{aligned}
\end{align*}
where
\begin{align*}
  \begin{aligned}
    \gamma_1 &= \frac{1}{1+\sum_{k=1}^{T-1}(1-p)^{k}\prod_{\ell=1}^k\sigma_\ell}
  \end{aligned}
  \qquad
  \begin{aligned}
    \gamma_t &= \frac{(1-p)^{t-1}\prod_{\ell=1}^{t-1}\sigma_\ell}{1+\sum_{k=1}^{T-1}(1-p)^{k}\prod_{\ell=1}^k\sigma_\ell} \quad \forall t > 1.
  \end{aligned}
\end{align*}
Here, I interpret $\gamma_t$ as the probability that the respondent is in period $t$ \textit{unconditional} on the offer received. The respondent's unconditional belief that he is in period $t$ can be interpreted as the \textit{long-run frequency} with which he enters period $t$, given the strategy profiles $(\sigma_1,...,\sigma_T)$ and $p$.\footnote{See \cite{ahp} or \cite{lambertmarpleshoham2019} for a discussion of these unconditional beliefs.}

Given any proposer strategy profile $\sigma=(\sigma_1,...,\sigma_T)$ and respondent strategy $p$, the proposer's value at time $t$ is
\begin{align*}
    U_t^P(\sigma,p)=\sigma_tpu_P^P+(1-\sigma_t)u_R^P+\delta\sigma_t(1-p)U_{t+1}^P(\sigma,p)
\end{align*}
where I define the final continuation value $U_{T+1}^P(\sigma,p)=0$. That is, $U_t^P(\sigma,p)$ is the proposer's expected utility under strategy profile $(\sigma,p)$ at period $t$. The first term represents the returns from the respondent accepting $a_P$, weighted by the probability of acceptance. The second term, similarly, is the payoff from offering $a_R$. The final term is the continuation value, weighted by the probability an offer is rejected in period $t$ and discounted by $\delta$. Similarly, the respondent's value function at time $t$ is
\begin{align*}
    U_t^R(\sigma,p)=\sigma_tpu_P^R+(1-\sigma_t)u_R^R+\delta \sigma_t (1-p)U_{t+1}^R(\sigma,p)
\end{align*}
where I define $U_{T+1}^R(\sigma,p)=0$. With this notation, I define my main solution concept.

\begin{defn}[MBE]
    A proposer strategy $\sigma^*$, a respondent strategy $p^*$, and beliefs $\alpha^P=(\alpha_1^P,...,\alpha_T^P)$ and $\alpha^R=(\alpha_1^R,...,\alpha_T^R)$ constitute a Multiselves Bargaining Equilibrium (MBE) if
    \begin{enumerate}[(i)]
        \item $\sigma^*$ is a best-response to $p^*$. That is, for each $t$,
        \begin{align*}
            \sigma_t^*\in\arg\max_{s_t\in[0,1]}s_tp^*u_P^P+(1-s_t)u_R^P+\delta s_t(1-p^*)U_{t+1}^P(\sigma^*,p^*)
        \end{align*}
        \item $\alpha^P$ and $\alpha^R$ are consistent with $\sigma^*$ and $p^*$.
        \item $p^*$ satisfies
        \begin{align*}
            p^*\in\arg\max_{p\in[0,1]}\sum_{t=1}^T\alpha_t^P\left(p\delta^{t-1}u_P^R+(1-p)\delta^tU_{t+1}^R(\sigma^*,p^*)\right)
        \end{align*}
        \item If $\exists t$ such that $\sigma_t^*>0$, then $p^*>0$.
    \end{enumerate}
\end{defn}

Condition (i) in the definition of MBE simply requires the proposer best respond to the respondent's strategy. Condition (ii) requires beliefs to be correct on the equilibrium path. Condition (iii) is more nuanced, and relies heavily on the interpretation of optimal strategies in the absentminded driver problem in \cite{ahp}. It requires that the respondent best responds after the proposer's preferred offer, holding fixed what they would do in a continuation of the game. \cite{ahp} argue that an absentminded decision maker cannot \textit{simultaneously} choose his action at each exit (absentminded players will forget that they have made a deviation). This manifests itself in condition (iii) since the continuation payoff $U_{t+1}^R$ depends on $p^*$, which the respondent takes as given. Condition (iv) is weak, and requires that the proposer never makes an offer that is rejected for sure. This condition is redundant whenever $\delta<1$,\footnote{That is, conditions (i)-(iii) imply (iv).} and eliminates an uninteresting class of behaviors when $\delta=1$.\footnote{In particular, it eliminates strategy profiles of the following form: $a_P$ is offered in some period $t<T$, rejected for sure, and $a_R$ is eventually accepted.}

\section{Analysis}\label{S:analysis}

One equilibrium is immediately obvious: $\sigma_t^*=1$ for all $t$ and $p^*=1$. In this equilibrium, there is no delay --- the proposer always makes her preferred offer and the respondent always accepts. There is no profitable deviation from any party: from the perspective of the respondent, rejecting an offer of $a_P$ leads to a payoff of $\delta u_P^R$, since $\alpha_t^P=1$ and $\sigma_2^*=1$. Accepting an offer of $a_P$ leads to a payoff $u_P^R$. This equilibrium has a natural analogue in the case of perfect recall. The unique subgame-perfect equilibrium of the version of this game with perfect recall is offering $a_P$ in each period and the respondent accepting immediately. 

There is another MBE which does \textit{not} have an analogue in the perfect recall version of this game (assuming $\delta\geq u_P^R/u_R^R)$: $\sigma_t^*=0$ for all $t$ and $p^*=0$.\footnote{In fact, as I demonstrate in Theorem \ref{T:char}, this is an equilibrium so long as $p^*\leq\frac{u_R^P(1-\delta)}{u_P^P-\delta u_R^P}$.} To support this as an equilibrium, set $\alpha_1^P=1$. After receiving an (off-path) offer $a_P$, the respondent forms a degenerate belief on $t=1$. Thus, rejecting the offer $a_P$ gives a payoff of $\delta u_R^R\geq u_P^R$. So the respondent optimally chooses $p^*=0$, and the proposer always offers $a_R$. In this equilibrium, there is also no delay: the proposer always makes the respondent's preferred offer and the respondent always accepts. Unlike the previous pure strategy MBE, this equilibrium has no natural analogue when parties have perfect recall. With perfect recall, the respondent is willing to accept any offer in period $T$, so the proposer should offer $a_P$ in each period by backward induction. With imperfect recall, the respondent may rationally believe that the proposers preferred offer is made in $t=1$.

In bargaining games with perfect recall and perfect information, the source of bargaining power is the \textit{recognition process}, or the process by which offers are made. When one party (the proposer) makes all the offers, the other has no bargaining power. Since the respondent never makes an offer, he will accept anything, and the proposer will exclusively offer $a_P$. With imperfect recall, however, the game can no longer be solved by backward induction. The respondent has bargaining power, since he can credibly reject an offer of $a_P$. The respondent is \textit{rationally confused} about the trade deadline, and thus optimally rejects $a_P$.

Finally, my main equilibrium of interest is a mixing equilibrium. I'll construct this equilibrium as follows: suppose $\sigma_T^*\in(0,1)$. Then, since the proposer must be indifferent between offering $a_P$ and $a_R$ at time $T$, $p^*u_P^P=u_R^P\implies p^*=u_R^P/u_P^P$. Next, examine the proposer's program at time $T-1$. This program is linear in $s_{T-1}$ with slope
\begin{align}\label{E:greedearly}
    p^*u_P^P-u_R^P+\delta(1-p^*)U_T^P(\sigma^*,p^*)=\delta(1-p^*)U_T^P(\sigma^*,p^*)>0,
\end{align}
so $\sigma_{T-1}^*=1$. Working backwards, one can show that for all $t<T$, $\sigma_t^*=1$.

Next, the respondent must be indifferent between accepting and rejecting the offer $a_P$. Observe that the respondent's program is linear in $p$ with slope
\begin{align*}
    \sum_{t=1}^T\alpha_t^P\left(\delta^{t-1}u_P^R-\delta^tU_{t+1}^R(\sigma^*,p^*)\right).
\end{align*}
In order for there to be an equilibrium with $p^*=u_R^P/u_P^P$, there must exist some $\sigma_T\in(0,1)$ which solves
\begin{align}\label{E:char}
    \sum_{t=1}^T\alpha_t^P(\sigma_T)\left(\delta^{t-1}u_P^R-\delta^tU_{t+1}^R(\sigma_1^*,...,\sigma_{T-1}^*,\sigma_T,p^*)\right)=0
\end{align}
where $\alpha_t^P(\sigma_T)$ is the unique conditional belief that the time period is $t$ consistent with $(\sigma_1^*,...,\sigma_{T-1}^*,\sigma_T)$ and $p^*$. In the proof of Theorem \ref{T:char}, I show that Equation \eqref{E:char} has a unique solution $\sigma_T^*\in(0,1)$ whenever $\delta$ exceeds some threshold $\underline{\delta}(T)\in(0,1)$ which depends on $T$.

\begin{thm}[Equilibrium Characterization]\label{T:char}
    There exists a threshold $u_P^R/u_R^R\leq\underline{\delta}(T)<1$ such that, if $T\geq2 $ and $\delta\geq\underline{\delta}(T)$, the following cases exhaust all possibilities of strategy profiles $(\sigma^*,p^*)$ that can occur in MBE:
    \begin{enumerate}
        \item \textit{(Proposer Preferred Equilibrium):} $\sigma_t^*=1$ for all $t$ and $p^*=1$.
        \item \textit{(Respondent Preferred Equilibrium):} $\sigma_t^*=0$ for all $t$ and $p^*\leq\frac{u_R^P(1-\delta)}{u_P^P-\delta u_R^P}$.
        \item \textit{(Mixing Equilibrium):} $\sigma_t^*=1$ for all $t<T$, $p^*=u_R^P/u_P^P$, and $\sigma_T^*$ uniquely solves Equation \eqref{E:char}.
    \end{enumerate}
\end{thm}
Theorem \ref{T:char} characterizes behavior that can occur in equilibrium. Appendix \ref{A:sequential} verifies that all MBE described in Theorem \ref{T:char} satisfy a natural \textit{sequential equilibrium} refinement.\footnote{See multiself sequential equilibria in \cite{lambertmarpleshoham2019}.}

Absentmindedness does not imply that delay \textit{must} occur. The proposer-preferred and respondent-preferred equilibria both feature immediate agreement, and all three equilibrium classes survive the sequential refinement in Appendix \ref{A:sequential}. Thus, absentmindedness creates an equilibrium channel for delay that is unavailable under perfect recall, rather than selecting delay uniquely. 

The respondent has some bargaining power in the mixing equilibrium. Since the respondent is unaware of the calendar time, he can rationally reject his least preferred offer with some probability even if he is at the trade deadline. The respondent's randomization disciplines the proposer to make favorable offers to the respondent in the final period of the game with some probability. %

In the mixing equilibrium, there is potentially substantial delay. The respondent strategy $p^*=u_R^P/u_P^P$ also does not depend on $T$ or $\delta$ (so long as $\delta\geq\underline{\delta}(T)$), so delay persists even as $T\to\infty$ and $\delta\to1$. The following corollary characterizes several immediate facts regarding the probability of delay, the probability a deal falls through, and the distribution over trade dates.

\begin{cor}\label{C:delay}
    If $T\geq 2$ and $\delta\geq\underline{\delta}(T)$, let $\hat{T}$ denote the date at which a deal occurs. Let $\hat{T}=\emptyset$ denote the event in which there is no trade. In the mixing MBE,
    \begin{enumerate}[(i)]
        \item $Pr(\hat{T}>1)=1-p^*=\frac{u_P^P-u_R^P}{u_P^P}$.
        \item $Pr(\hat{T}=t)=p^*(1-p^*)^{t-1}=\frac{u_R^P(u_P^P-u_R^P)^{t-1}}{(u_P^P)^t}$ for all $t<T$.
        \item $Pr(\hat{T}=T)=(1-p^*)^{T-1}(\sigma_T^*p^*+(1-\sigma_T^*))$.
        \item $Pr(\hat{T}=\emptyset)=\sigma_T^*(1-p^*)^T$.
        \item The expected date at which a deal occurs, given that a deal is made, is
        \begin{align*}
            E[\hat{T}|\hat{T}\neq\emptyset]&=\left(\frac{p^*}{1-\sigma_T^*(1-p^*)^T}\right)\sum_{t=1}^T(1-p^*)^{t-1}t\\
        \end{align*}
    \end{enumerate}
\end{cor}
Observe that the probability a deal is made in period $t<T$ does not depend on the discount rate $\delta$ or the time horizon $T$. The probability the deal falls through, $Pr(\hat{T}=\emptyset)$, and the expected date of agreement depend on $T$ and $\delta$, however. The next result characterizes the manner in which these values depend on $\delta$ and $T$.
\begin{cor}\label{C:compstats}
If $T\geq 2$ and $\delta\geq \underline{\delta}(T)$, in the mixing MBE,
\begin{enumerate}[(i)]

    \item $Pr(\hat{T}=\emptyset)$ is increasing in $\delta$ and $E[\hat{T}|\hat{T}\neq\emptyset]$ is increasing in $\delta$.
    \item As $T\to\infty$, $Pr(\hat{T}=\emptyset)\to0$ and
    \begin{align*}
        E[\hat{T}|\hat{T}\neq\emptyset]\to\frac{1}{p^*}=\frac{u_P^P}{u_R^P}>1
    \end{align*}
\end{enumerate}
\end{cor}

As the players become more patient, the probability with which a deal falls through increases. Holding $\sigma$ fixed, increasing $\delta$ gives the respondent a strict incentive to reject an offer of $a_P$. To keep the respondent indifferent between accepting and rejecting $a_P$, $\sigma_T^*$ must increase as $\delta$ increases. But since the probability of no-deal is $Pr(\hat{T}=\emptyset)=\sigma_T^*(1-p^*)^T$ and $\sigma_T^*$ is increasing in $\delta$, $Pr(\hat{T}=\emptyset)$ must also be increasing in $\delta$. Moreover, conditional on reaching a deal, more patient players delay more. As the exogenous trade deadline $T\to\infty$, the probability of no-deal goes to $0$ as the respondent is likely to accept an offer prior to the distant deadline. The expected trade date as $T\to\infty$ is, as one might expect, finite.

Theorem \ref{T:char} can also be used to see why the restriction to Markov strategies for the proposer is without loss of generality without side payments. Suppose I do not restrict to Markov strategies. Regardless of the history of the game, in the final period the proposer either offers $a_P$, $a_R$, or they mix. If they send their preferred offer in period $T$ on the equilibrium path, it is a strict best response to send their preferred offer in any earlier period $t<T$. Similarly, if they send the respondent's preferred offer in the final period on the equilibrium path, it is a strict best response to send the respondent's preferred offer in any earlier period (provided that $p^*$ satisfies the bound in Theorem \ref{T:char}). Finally, if the proposer mixes in the final period on path, then it is a strict best response to offer $a_P$ in every period $t<T$. Therefore, Theorem \ref{T:char} describes the set of MBE outcomes even when one allows the proposer to play non-Markov strategies. This reasoning fails when side payments are allowed, as I will demonstrate in Appendix \ref{S:general}.

\subsection{Ex-Ante Equilibrium}

I define an alternate solution concept which is analogous to the concept of \textit{planning optimality} in \cite{ahp}. First, note that a strategy profile $\sigma$ for the proposer induces a distribution over pure strategy vectors denoted $\rho_\sigma\in\Delta(\{a_P,a_R\}^T)$.

\begin{defn}[Ex-Ante Equilibrium]
    A proposer strategy $\sigma^*$ and a respondent strategy $p^*$ is an \textit{ex-ante equilibrium} if
    \begin{enumerate}[(i)]
        \item $\sigma^*$ is a best response to $p^*$. That is, for each $t$,
        \begin{align*}
            \sigma_t^*\in\arg\max_{s_t\in[0,1]}s_tp^*u_P^P+(1-s_t)u_R^P+\delta s_t(1-p^*)U_{t+1}^P(\sigma^*,p^*)
        \end{align*}
        \item $p^*$ solves
        \begin{align}\label{E:exante}
            \max_p\sum_{\boldsymbol{a}\in\{a_P,a_R\}^T}\rho_{\sigma^*}(\boldsymbol{a})\sum_{t=1}^T\delta^{t-1}(1-p)^{t-1}\left(\mathbb{1}_{a_t=a_P}pu_P^R+\mathbb{1}_{a_t=a_R}u_R^R\right)\mathbb{1}_{a_i=a_P\;\forall i<t}
        \end{align}
    \end{enumerate}
\end{defn}

Equation \eqref{E:exante} has the following interpretation: each action profile $\boldsymbol{a}\in\supp\sigma^*$ induces an absentminded driver decision problem, with payoffs determined by whether or not the proposer offers $a_P$ or $a_R$ at each node. The respondent ex-ante commits to a probability of accepting $a_P$.

Next, observe that, if $\sigma_T^*>0$, then $\sigma_{T-1}^*=1$ by Equation \eqref{E:greedearly}. Applying this considerably simplifies the program in Equation \eqref{E:exante}:
\begin{align*}
    \max_p\sigma_T^*&\left(\sum_{t=1}^{T-1}\delta^{t-1}p(1-p)^{t-1}u_P^R+\delta^{T-1}(1-p)^{T-1}pu_P^R\right)\\&+(1-\sigma_T^*)\left(\sum_{t=1}^{T-1}\delta^{t-1}p(1-p)^{t-1}u_P^R+\delta^{T-1}(1-p)^{T-1}u_R^R\right).
\end{align*}
Essentially, the respondent faces an absentminded driver problem, but their decision problem has a random element. In particular, with probability $\sigma_T^*$, the payoffs at the final exit are $\delta^Tu_P^R$ for the respondent. With probability $1-\sigma_T^*$, the payoffs from the final exit are $\delta^Tu_R^R$, but the respondent can \textit{identify} the final exit (since it is the only exit where she receives her preferred offer in equilibrium).\footnote{In the absentminded driver problem, imagine the final exit has a sign indicating there are no further exits down the road.} Since the driver does not know which decision problem she faces, she selects an ex-ante optimal (or planning-optimal in the language of \cite{ahp}) exit probability $p^*$ to commit to. Figure \ref{fig:absent-minded} displays the problem in Equation \ref{E:exante} when $T=3$.

\begin{figure}[htbp]
    \centering

    \begin{minipage}{0.45\textwidth}
    \centering
    \begin{tikzpicture}[>=stealth, thick, scale=1.2, every node/.style={scale=1}]
    \node at (0,4.2) {START};

    \node (t1) at (0,3.5) {}; \draw[fill=black] (t1) circle (2pt); \node[left, xshift=-3pt] at (t1) {$t=1$};
    \node (t2) at (0,2.5) {}; \draw[fill=black] (t2) circle (2pt); \node[left, xshift=-3pt] at (t2) {$t=2$};
    \node (t3) at (0,1.5) {}; \draw[fill=black] (t3) circle (2pt); \node[left, xshift=-3pt] at (t3) {$t=3$};

    \draw[-] (0,4) -- (t1); %
    \draw[-] (t1) -- (t2);  %
    \draw[-] (t2) -- (t3);  %
    \draw[-] (t3) -- (0,0.5); %

    \draw (t1) -- (2,3.5); \node[above] at (1,3.5) {$A_P$}; \node at (2.3,3.5) {$u_P^R$};
    \draw (t2) -- (2,2.5); \node[above] at (1,2.5) {$A_P$}; \node at (2.4,2.5) {$\delta u_P^R$};
    \draw (t3) -- (2,1.5); \node[above] at (1,1.5) {$A_P$}; \node at (2.3,1.5) {$\delta^2 u_P^R$};

    \node at (0,0) {$0$};

    \node[left] at (0,3) {R};
    \node[left] at (0,2) {R};
    \node[left] at (0,0.85) {R};

    \node[draw, dashed, fit=(t1)(t2)(t3), inner sep=10pt] {};
    \end{tikzpicture}
    
    \vspace{0.3em}
    \small{$a_3 = a_P$}
    \end{minipage}
    \hspace{1cm}
    \begin{minipage}{0.45\textwidth}
    \centering
    \begin{tikzpicture}[>=stealth, thick, scale=1.2, every node/.style={scale=1}]
    \node at (0,4.2) {START};

    \node (t1) at (0,3.5) {}; \draw[fill=black] (t1) circle (2pt); \node[left, xshift=-3pt] at (t1) {$t=1$};
    \node (t2) at (0,2.5) {}; \draw[fill=black] (t2) circle (2pt); \node[left, xshift=-3pt] at (t2) {$t=2$};
    \node (t3) at (0,1.5) {}; \draw[fill=black] (t3) circle (2pt); \node[left, xshift=-3pt] at (t3) {$t=3$};

    \draw[-] (0,4) -- (t1); %
    \draw[-] (t1) -- (t2);  %
    \draw[-] (t2) -- (t3);  %
    \draw[-] (t3) -- (0,0.5); %

    \draw (t1) -- (2,3.5); \node[above] at (1,3.5) {$A_P$}; \node at (2.3,3.5) {$u_P^R$};
    \draw (t2) -- (2,2.5); \node[above] at (1,2.5) {$A_P$}; \node at (2.4,2.5) {$\delta u_P^R$};
    \draw (t3) -- (2,1.5); \node[above] at (1,1.5) {$A_R$}; \node at (2.3,1.5) {$\delta^2 u_R^R$};

    \node at (0,0) {$0$};

    \node[left] at (0,3) {R};
    \node[left] at (0,1.9125) {R};
    \node[left] at (0,1) {R};

    \node[draw, dashed, fit=(t1)(t2), inner sep=10pt] {};
    \end{tikzpicture}

    \vspace{0.3em}
    \small{$a_3 = a_R$}
    \end{minipage}

    \caption{Decision problems in the case of $T=3$ when $a_3=a_P$ and $a_3=a_R$.}
    \label{fig:absent-minded}
\end{figure}

When $\sigma_T^*\in(0,1)$, it must be the case that $p_G^*=u_R^P/u_P^P$. If $p_G^*=u_R^P/u_P^P$ solves the problem in Equation \eqref{E:exante}, one obtains Equation \eqref{E:char} as a necessary and sufficient condition. Thus, the ex-ante best-response of the absentminded respondent to a strategy $\sigma^*$ is the same as the optimal strategy at any information set. This result is analogous to the fact that every \textit{planning-optimal} strategy is also \textit{action-optimal} in \cite{ahp}.

\begin{thm}[Equivalence between MBE and Ex-Ante Equilibria]\label{T:exanteequiv}
    Let $(\sigma^*,p^*)$ be a strategy profile. Then $(\sigma^*,p^*)$ can occur in an MBE if and only if $(\sigma^*,p^*)$ is an ex-ante equilibrium.
\end{thm}

The equivalence between MBE and ex-ante equilibria is also convenient; finding ex-ante equilibria does not require one to specify the agent's conditional beliefs over calendar time.\footnote{In particular, this equivalence is used in the proof of Theorem \ref{T:char}.} The respondent's decision problem is a simple optimization problem depending only on $\sigma^*$. This equivalence is used extensively in the proof of Theorem \ref{T:char}.

\subsection{Interpretation and Application}

Theorem \ref{T:exanteequiv} provides a useful interpretation of absentmindedness in the bargaining game. An absentminded respondent need not literally forget the calendar date. The theorem shows that the same equilibrium behavior arises in a related game where the respondent chooses a stationary behavioral strategy ex-ante and applies it throughout the negotiation. Absentmindedness therefore captures a form of coarse contingent planning: the respondent cannot finely tailor his acceptance decision to the number of bargaining opportunities that remain. This interpretation is particularly natural in dynamic games with long time horizons, where histories are long and complex and backward induction may not be behaviorally compelling.

Coarse contingent planning can also arise from organizational constraints. Consider, for instance, a collective bargaining committee operating under a standing mandate. The mandate may specify which proposals the committee can accept without specifying a different acceptance rule for every possible stage of negotiations. The committee can therefore behave as if it were absentminded even if its members know exactly how much calendar time remains. If the opposing party understands the mandate, this rigidity becomes a source of commitment: an unfavorable offer may be rejected even close to a deadline, inducing the other party to concede rather than risk a breakdown. The same logic can apply whenever a bargaining representative operates under a relatively coarse approval rule.

The methodology outlined above can be applied to a larger class of dynamic games. Bargaining is a natural application for two reasons. First, upon rejecting an offer in period $t$, the situation in period $t+1$ looks identical, except that the remaining number of periods has decreased by $1$. Importantly, the set of actions available to parties at each decision node remains the same (e.g. the respondent chooses whether to accept or reject). Such games are well-suited to the absentmindedness approach, since all decision nodes for a player can belong to the same information set. The classic centipede game also has this property (Figure \ref{F:cent}). Contrast this, for instance, with chess. In chess, after a player makes a move, the board state when a player moves again is almost always different, and players clearly understand this. The set of possible actions, or moves, also changes in the future. Second, the predictions under imperfect recall and under perfect recall in bargaining games are substantially different.

\begin{figure}[h!]
    \centering
    \begin{tikzpicture}[
    scale=1.5,
    font=\small,
    dot/.style={circle, fill=black, inner sep=1.5pt}
]

\node[dot, label={above left:P1}] (n1) at (0,0) {};
\node[dot, label={above left:P2}] (n2) at (2,0) {};
\node[dot, label={above left:P1}] (n3) at (4,0) {};
\node[dot, label={above left:P2}] (n4) at (6,0) {};

\coordinate (end) at (8,0);

\draw (n1) -- (n2) node[midway, above] {$C$};
\draw (n2) -- (n3) node[midway, above] {$C$};
\draw (n3) -- (n4) node[midway, above] {$C$};
\draw (n4) -- (end) node[midway, above] {$C$};

\node[right] at (end) {$(4,4)$};

\draw (n1) -- (0,-1.2) node[midway, left] {$S$} node[below] {$(1,1)$};
\draw (n2) -- (2,-1.2) node[midway, left] {$S$} node[below] {$(0,3)$};
\draw (n3) -- (4,-1.2) node[midway, left] {$S$} node[below] {$(3,2)$};
\draw (n4) -- (6,-1.2) node[midway, left] {$S$} node[below] {$(2,5)$};

\end{tikzpicture}
    \caption{The classic centipede game \citep{rosenthal81}}
    \label{F:cent}
\end{figure}

\section{Extensions}\label{S:extensions}

The baseline model makes two simplifying assumptions: the absentminded player responds to offers, and the parties bargain over a binary course of action without side payments. Neither assumption is essential for absentmindedness to generate delay, but both affect the set of equilibria that can be sustained. I briefly discuss each extension here; the formal analysis is contained in the Supplemental Material.

\subsection{Absentminded Proposer}

Often, results in bargaining are sensitive to the bargaining protocol --- that is, they depend on who makes offers and when. Appendix \ref{S:propforg} therefore considers the case where the proposer is absentminded and the respondent has perfect recall. For ease of exposition, I focus on $T=2$. As in the baseline model, there is a proposer-preferred equilibrium with immediate agreement, a respondent-preferred equilibrium when players are sufficiently patient, and, for some parameter values, a mixed equilibrium with delay. Proposition \ref{P:forgetfulchar} characterizes these equilibria.

The mechanism behind delay is similar to the baseline model. An absentminded proposer cannot condition her offer on calendar time. In a mixed equilibrium, she sometimes makes her preferred offer $a_P$ even when it will be rejected in period $1$, because the same offer is accepted by the cognizant respondent in period $2$. The respondent, in turn, may reject $a_P$ in period $1$ in hopes of receiving $a_R$ after the proposer moves again. Hence, imperfect recall can generate delay even when the absentminded party is the party making offers.

There is an important difference from the baseline model. When the respondent is absentminded, he may reject $a_P$ even at the trade deadline. The proposer therefore faces a genuine risk that the deal falls through, and may concede by offering $a_R$ in the final period. When the proposer is absentminded, the respondent has perfect recall and always accepts $a_P$ in the final period. Thus, there is no risk that the deal falls through in the mixed equilibrium. This makes delay harder to sustain when players are very patient. In particular, the mixed equilibrium exists only for intermediate discount factors. If the proposer is sufficiently patient, she is willing to repeatedly make her preferred offer because she knows that it will eventually be accepted. The contrast isolates the role played by the trade deadline in the baseline model: absentmindedness gives the respondent bargaining power precisely because he cannot be disciplined by a last-minute ultimatum.

\subsection{Side Payments}

Appendix \ref{S:general} allows the parties to bargain over how to divide a fixed surplus. The proposer offers a share $x\in[0,1]$ to the respondent, so the offer space is now continuous. This additional flexibility changes the analysis considerably. In the baseline model, the proposer can only choose between $a_P$ and $a_R$. With side payments, she can make arbitrarily small concessions, which creates an undercutting problem.

This point is clearest when the proposer follows a Markov strategy. Proposition \ref{P:acceptany} shows that the period-$2$ offer distribution must be degenerate. Proposition \ref{P:cont} then shows that, if $\delta<1$, the proposer offers $x=0$ in both periods and the respondent accepts immediately. Intuitively, suppose the proposer planned to make some offer $x>0$ in the future. She can instead offer approximately $\delta x$ immediately. The respondent is willing to accept because waiting for $x$ is discounted, while the proposer strictly prefers to concede less. The ability to fine-tune offers therefore eliminates the bargaining power generated by absentmindedness in a Markov equilibrium.

Delay can reappear when the parties are patient and the proposer is allowed to use history-dependent strategies. Proposition \ref{P:contdelay} shows that an MBE with delay exists if and only if $\delta=1$. The construction closely resembles the baseline model. The proposer makes a ``greedy'' offer $x_L$ in the first period and, after rejection, mixes between $x_L$ and a more favorable offer $x_H$ in the second period. Off the equilibrium path, offers below $x_H$ are rejected because the proposer responds to their rejection by offering $x_H$ in the next period. From the proposer's perspective, these continuation strategies effectively constrain the viable offer set to the binary set $\{x_L,x_H\}$. The offer $x_L$ therefore plays the role of $a_P$ in the baseline model, while $x_H$ plays the role of $a_R$.

This binary structure is not simply an artifact of the construction. Proposition \ref{P:necessity} shows that, among equilibria with delay satisfying an on-path Markov property, the proposer makes a single offer $x_L$ in period $1$ and mixes only between $x_L$ and some $x_H>x_L$ in period $2$. Moreover, the more favorable offer $x_H$ is the only on-path offer which reveals that the trade deadline has arrived. Thus, when delay survives with side payments, equilibrium behavior recreates the basic structure of the baseline model: the proposer repeatedly demands favorable terms, the absentminded respondent sometimes rejects in hopes of a concession, and the proposer may make a discrete concession as the deadline looms.

Taken together, these extensions suggest that the central mechanism is not tied to a particular recognition process or to a literal binary offer space. What matters is that the absentminded party cannot finely condition behavior on the stage of the negotiation, and that the opposing party cannot costlessly undo the resulting rigidity by making an arbitrarily small concession. Changing who is absentminded affects the risk of breakdown and therefore the range of parameters over which delay can be sustained; allowing side payments makes the richness of the offer space important. In both cases, however, delay arises from the same basic force as in the baseline model: coarse contingent planning weakens the force of backward induction and can create bargaining power.

\section{Conclusion}

Absentmindedness can be a source of delay in bargaining. An absentminded respondent rejects unfavorable offers in hopes of receiving a favorable offer later; a cognizant proposer makes favorable offers more frequently as the game goes on to prevent the deal from falling through. The key difference from standard, perfect recall, finite-horizon bargaining models is that the absentminded party can rationally reject unfavorable offers, even in the final period of the game. Inefficiency is a natural consequence of this relationship, as the cognizant party makes offers which are likely to be rejected in hopes of securing an agreement on her preferred terms. Delay is persistent; remarkably, in the patient limit and as the deadline becomes increasingly distant, the probability of delay remains constant. This analysis extends to the case with transferable utility, provided that the parties are patient. More broadly, the results identify coarse contingent planning as a source of commitment in bargaining. A player who cannot finely tailor his behavior to every stage of a negotiation may be harder to discipline with a last-minute ultimatum, and this rigidity can translate into bargaining power.

\appendix
\section*{Appendix}
\begin{lem}\label{L:equiv}
    Equation \eqref{E:char} has a solution $\sigma_T^*\in(0,1)$ if and only if $p^*=u_R^P/u_P^P$ maximizes Equation \eqref{E:exante} given $\sigma_T^*$, and given $\sigma_t^*=1$ for $t<T$.
\end{lem}
\begin{proof}
Equation (3) can be written as
\begin{align*}
    \max_{p} \sigma_{T}^{*} &\left( \sum_{t=1}^{T-1} \delta^{t-1} p (1-p)^{t-1} u_P^R + \delta^{T} (1-p)^{T-1} p u_P^R \right)\\&+(1-\sigma_{T}^{*}) \left( \sum_{t=1}^{T-1} \delta^{t-1} p (1-p)^{t-1} u_P^R + \delta^{T-1} (1-p)^{T-1} u_R^R \right).
\end{align*}
Taking first order conditions with respect to $p$ yields
\begin{align*}
    \sigma_{T}^{*} &\left( \sum_{t=1}^{T} \delta^{t-1} u_P^R \left( (1-p)^{t-1} - p(t-1)(1-p)^{t-2} \right) \right)\\&+(1-\sigma_{T}^{*}) \left( \sum_{t=1}^{T-1} \delta^{t-1} u_P^R \left( (1-p)^{t-1} - p(t-1)(1-p)^{t-2} \right) - \delta^{T-1} u_R^R (T-1)(1-p)^{T-2} \right) = 0.
\end{align*}
Substituting unconditional beliefs $\gamma_t$ into the first order condition yields
\begin{align*}
    \sigma_{T}^{*}& \left( \sum_{t=1}^{T} \delta^{t-1} u_P^R \left( \gamma_{t} - \frac{p}{1-p}(t-1)\gamma_{t} \right) \right)\\&+(1-\sigma_{T}^{*}) \left( \sum_{t=1}^{T-1} \delta^{t-1} u_P^R \left( \gamma_{t} - \frac{p}{1-p}(t-1)\gamma_{t} \right) - \delta^{T-1} u_R^R (T-1)\gamma_{T-1} \frac{1}{1-p} \right) = 0.
\end{align*}
Using the definition of conditional beliefs $\alpha_t^P$, where $\alpha_{t}^P = \frac{\sigma_t^* \gamma_t}{\sum \sigma_k^* \gamma_k}$, one can rewrite the above expression as
\begin{align*}
    \sum_{t=1}^{T-1} \delta^{t-1} &u_P^R \left( \alpha_{t}^P - \frac{p}{1-p}(t-1)\alpha_{t}^P \right) + \sigma_{T}^{*} \delta^{T-1} u_P^R \alpha_{T}^P\\&=(1-\sigma_{T}^{*}) \delta^{T-1} u_R^R (T-1) \frac{\alpha_{T-1}^P}{1-p} + \sigma_{T}^{*} u_P^R \frac{p}{1-p}(T-1) \alpha_{T-1}^P \delta^{T-1}.
\end{align*}
It follows that
\begin{align*}
    \sum_{t=1}^{T} &\alpha_{t}^P \delta^{t-1} u_P^R - \sum_{t=1}^{T-1} \alpha_{t}^P \delta^{t-1} u_P^R (t-1) \frac{p}{1-p}\\&=(1-\sigma_{T}^{*}) \delta^{T-1} u_R^R (T-1) \frac{\alpha_{T-1}^P}{1-p} + \sigma_{T}^{*} u_P^R \frac{p}{1-p}(T-1) \alpha_{T-1}^P \delta^{T-1}.
\end{align*}
Expanding the summation terms and plugging in $p=p^{*}$ yields the following expression
\begin{align*}
    \sum_{t=1}^{T} \alpha_{t}^P \left( \delta^{t-1} u_P^R - \delta^{t} \left( \sum_{l=1}^{T-t-1} \delta^{l-1} (1-p^{*})^{l-1} p^{*} u_P^R + \delta^{T-t}(1-p^{*})^{T-t} (\sigma_{T}^{*} p^{*} u_P^R + (1-\sigma_{T}^{*}) u_R^R) \right) \right) 
\end{align*}
\begin{align*}
     =\sum_{t=1}^{T} \alpha_{t}^P \left( \delta^{t-1} u_P^R - \delta^{t} U_{t+1}^{R}(\sigma^{*}, p^{*}) \right) = 0.
\end{align*}
Therefore, $\sigma_{T}^{*}$ solves Equation \eqref{E:char} if and only if $p^{*}$ maximizes Equation \eqref{E:exante} given $\sigma_{T}^{*}$.
\end{proof}

\begin{proofof}{\bf Theorem \ref{T:char}} Let $(\sigma^*,p^*)$ be an MBE. I consider three exhaustive cases: $\sigma_T^*=1$, $\sigma_T^*=0$, and $\sigma_T^*\in(0,1)$. I show that these three cases correspond to cases 1, 2, and 3 in the statement of Theorem \ref{T:char}.

\noindent\textbf{Case I ($\sigma_T^*=1)$:} If $\sigma_T^*=1$, by the proposer's program at period $T$,
\begin{align*}
    p^*u_P^P\geq u_R^P\implies p^*\geq\frac{u_R^P}{u_P^P}.
\end{align*}
Then, by Equation \eqref{E:greedearly}, $\sigma_t^*=1$ for all $t$. Then, the respondent continuation payoff can be expressed as follows:
\begin{align*}
    U_t^R(\sigma^*,p)=pu_P^R+\delta(1-p)U_{t+1}^R(\sigma^*,p_G).
\end{align*}
Iterating backwards from $t=T$ yields a simple closed form solution for $U_t^R(\sigma^*,p)$:
\begin{align*}
    U_{t}^R(\sigma^*,p)=pu_P^R+(1-p)\left(\sum_{k=1}^{T-t}\delta^k(1-p)^{k-1}pu_P^R\right)\leq u_P^R.
\end{align*} 
It follows that the solution to condition (iii) of the definition of MBE is $p^*=1$. This is the proposer preferred.

\noindent\textbf{Case II ($\sigma_{T}^*=0)$:} If $\sigma_T^*=0$, by the proposer's program at period $T$,
\begin{align*}
    p^*u_P^P\leq u_R^P\implies p^*\leq\frac{u_R^P}{u_P^P},
\end{align*}
which implies that
\begin{align*}
    p^*u_P^P+\delta(1-p^*)u_R^P\leq u_R^P+\delta\left(1-\frac{u_R^P}{u_P^P}\right)u_R^P<u_R^P.
\end{align*}
Therefore, $\sigma_{T-1}^*=0$ (the proposer's program in $T-1$ is linear in $s_t$ with slope $p^*u_P^P+\delta(1-p^*)u_R^P-u_R^P$). Iterating backwards yields $\sigma_t^*=0$ for all $t$.

To support this as an equilibrium, it must be the case that $s_t=0$ solves the proposer's program in each period $t$. That is,
\begin{align*}
    p^*u_P^P-u_R^P+\delta(1-p^*)u_R^P&\leq0\\\iff p^*\leq\frac{u_R^P(1-\delta)}{u_P^P-\delta u_R^P}.
\end{align*}
These values of $p^*$ can be supported in equilibrium by setting $\alpha_1^P=1$ and observing that $p^*=0$ is a best response to receiving a greedy offer since
\begin{align*}
    u_P^R\leq\delta u_R^R
\end{align*}
whenever $\delta\geq u_P^R/u_R^R$.

\noindent\textbf{Case III $(\sigma_T\in(0,1))$:} By Lemma \ref{L:equiv}, Equation \eqref{E:char} has a solution if and only if Equation \eqref{E:exante} is maximized at $p^*=u_R^P/u_P^P$. Define the following values
\begin{align*}
    A_\delta(p)&=\sum_{t=1}^T\delta^{t-1}p(1-p)^{t-1}u_P^R\\B_\delta(p)&=\sum_{t=1}^{T-1}\delta^{t-1}p(1-p)^{t-1}u_P^R+\delta^{T-1}(1-p)^{T-1}u_R^R.
\end{align*} 
Then the objective in Equation \eqref{E:exante} can be written as
\begin{align*}
    \sigma_T^*A_\delta(p)+(1-\sigma_T^*)B_\delta(p)
\end{align*}
Therefore $p=p^*=u_R^P/u_P^P$ solves Equation \eqref{E:exante} if the objective is concave in $p$ and if the first order condition holds. Since we must have
\begin{align*}
    \sigma_T^*(\delta)=\frac{-B_\delta'(p^*)}{A_\delta'(p^*)-B_\delta'(p^*)}
\end{align*}
in order for the first order condition to hold, we must have that
\begin{align*}
    \frac{-B'(p^*)}{A'(p^*)-B'(p^*)}\in(0,1).
\end{align*}
Observe first that
\begin{align*}
    A_\delta(p)&=u_P^Rp\frac{(1-\delta^T(1-p)^T)}{1-\delta(1-p)}\\B_\delta(p)&=u_P^Rp\frac{(1-\delta^{T-1}(1-p)^{T-1})}{1-\delta(1-p)}+u_R^R\delta^{T-1}(1-p)^{T-1}.
\end{align*}
and observe that
\begin{align*}
    A_\delta'(p)&=u_P^R\frac{(1-\delta^T(1-p)^T)}{1-\delta(1-p)}\\&+u_P^Rp\frac{(1-\delta(1-p))(\delta^TT(1-p)^{T-1})-(1-\delta^T(1-p)^T)\delta}{(1-\delta(1-p))^2}.
\end{align*}
Since $A_\delta'(p)>0$, then $-B_\delta'(p^*)<-B'_\delta(p^*)+A_\delta'(p^*)$ so $\sigma_T^*(\delta)<1$. Moreover, $\exists \epsilon_1>0$ such that $\sigma_T^*(\delta)>0$ for all $\delta\in[1-\epsilon_1,1]$ if and only if $\sigma_T^*(1)>0$, by continuity. Observe that
\begin{align*}
    A_1(p)&=u_P^R(1-(1-p)^T)\\B_1(p)&=u_P^R(1-(1-p)^{T-1})+u_R^R(1-p)^{T-1}.
\end{align*}
To show that $\sigma_T^*(1)>0$, it suffices to show that $-B_1'(p^*)>0\iff B_1'(p^*)<0$. Observe that
\begin{align*}
    B_1'(p^*)&=u_P^R(T-1)(1-p^*)^{T-2}-u_R^R(T-1)(1-p^*)^{T-2}\\&=-(u_R^R-u_P^R)(T-1)(1-p^*)^{T-2}<0
\end{align*}
since $u_R^R>u_P^R$, $T>1$, and $p^*<1$.

I next show that, for sufficiently high $\delta$, the objective is uniquely maximized at $p^*$. Define
\begin{align*}
    F_\delta(p)=\sigma_T^*(\delta)A_\delta(p)+(1-\sigma_T^*(\delta))B_\delta(p).
\end{align*}
At $\delta=1$,
\begin{align*}
    F_1'(p)=(1-p)^{T-2}\left(\sigma_T^*(1)u_P^RT(1-p)-(1-\sigma_T^*(1))(u_R^R-u_P^R)(T-1)\right).
\end{align*}
Since $F_1'(p^*)=0$, this expression can be written as
\begin{align*}
    F_1'(p)=\sigma_T^*(1)u_P^RT(1-p)^{T-2}(p^*-p).
\end{align*}
Thus, $F_1'(p)>0$ for $p<p^*$ and $F_1'(p)<0$ for $p>p^*$, so $p^*$ is the unique global maximizer of $F_1$. Moreover,
\begin{align*}
    F_1''(p^*)=-\sigma_T^*(1)u_P^RT(1-p^*)^{T-2}<0.
\end{align*}

The function $F_\delta(p)$ and its derivatives are continuous in $(p,\delta)$, and $\sigma_T^*(\delta)$ is continuous for $\delta$ sufficiently close to $1$. Since $p^*$ is the unique global maximizer of $F_1$, there exists a neighborhood $N$ of $p^*$ and $\eta>0$ such that $F_1(p^*)-F_1(p)>2\eta$ for all $p\notin N$. By uniform continuity, this inequality continues to hold with $\eta$ in place of $2\eta$ for all $\delta$ sufficiently close to $1$. Moreover, by the strict second-order condition at $p^*$, I can choose $N$ and $\epsilon_2>0$ such that $F_\delta''(p)<0$ for all $p\in N$ and $\delta\in[1-\epsilon_2,1]$. Since $F_\delta'(p^*)=0$ by construction, $p^*$ is then the unique maximizer of $F_\delta$ on $N$, and hence the unique global maximizer.

Therefore, $\sigma_T^*=\frac{-B_\delta'(p^*)}{A_\delta'(p^*)-B_\delta'(p^*)}$ uniquely supports the mixing equilibrium whenever $\delta\geq\max\{1-\epsilon_1,1-\epsilon_2\}$. By Lemma \ref{L:equiv}, $(\sigma^*,p^*)$ is the unique MBE with $\sigma_T^*\in(0,1)$. Thus, the threshold can be defined as $\underline{\delta}(T)=\max\{1-\epsilon_1,1-\epsilon_2,u_P^R/u_R^R\}$.

\end{proofof}

\begin{proofof}{\bf Corollary \ref{C:delay}}
    Throughout, let $\sigma^*$ denote the mixed MBE from Theorem \ref{T:char}.

    \begin{enumerate}[(i)]
        \item Observe that
        \begin{align*}
            Pr(\hat{T}>1)=1-Pr(\hat{T}=0)=1-(\sigma_1^*p^*+(1-\sigma_1^*))=1-p^*
        \end{align*}
        since $\sigma_1^*=1$, as desired.
        \item Observe that
        \begin{align*}
            Pr(\hat{T}=t)&=Pr(\hat{T}\neq1,...,t-1)Pr(\hat{T}=t|\hat{T}\neq1,...,t-1)\\&=(1-p^*)^{t-1}p^*
        \end{align*}
        as desired.
        \item Observe that
        \begin{align*}
            Pr(\hat{T}=T)&=Pr(\hat{T}\neq1,...,T-1)Pr(\hat{T}=t|\hat{T}\neq1,...,T-1)\\&=(1-p_G^*)^{t-1}(\sigma_T^*p_G^*+(1-\sigma_T^*))
        \end{align*}
        as desired.
        \item Observe that
        \begin{align*}
            Pr(\hat{T}=\emptyset)&=(1-p^*)^{T-1}(1-\sigma_T^*p^*-(1-\sigma_T^*))\\&=\sigma_T^*(1-p^*)^{T}
        \end{align*}
        as desired.
        \item Observe that
        \begin{align*}
            Pr(\hat{T}=t|\hat{T}\neq\emptyset)=\frac{Pr(\hat{T}=t)}{Pr(\hat{T}\neq\emptyset)}=\frac{p^*(1-p^*)^{t-1}}{1-\sigma_T^*(1-p^*)^T}
        \end{align*}
        and so it follows that
        \begin{align*}
            E[\hat{T}|\hat{T}\neq\emptyset]=\sum_{t=1}^TPr(\hat{T}=t|\hat{T}\neq\emptyset)t=\left(\frac{p_G^*}{1-\sigma_T^*(1-p_G^*)^T}\right)\sum_{t=1}^T(1-p_G^*)^{t-1}t
        \end{align*}
        as desired.
    \end{enumerate}
\end{proofof}

\begin{proofof}{\bf Corollary \ref{C:compstats}}
    Throughout, let $\delta>\delta'\geq\underline{\delta}(T)$. Let $\hat{T}_\delta$ and $\hat{T}_{\delta'}$ be the random variables denoting the trade date when the discount factor is $\delta$ and $\delta'$, respectively.
    \begin{enumerate}[(i)]
        \item Let
        \begin{align*}
            L(\sigma_T,\delta)=\sum_{t=1}^T\alpha_t^P(\sigma_T)\left(\delta^{t-1}u_P^R-\delta^t U_{t+1}^R(\sigma_1^*,...,\sigma_{T-1}^*,\sigma_T,p^*)\right)
        \end{align*}
        and observe that $\sigma_T^*(\delta)$ satisfies $L(\sigma_T^*(\delta),\delta)=0$. Observe that
        \begin{align*}
            L_1(\sigma_T^*(\delta),\delta)=\frac{-\sum_{k=1}^{T-1}\gamma_k\delta^k\delta^{k+1}(1-p^*)^{k+1}(p^*u_P^R-u_R^R)+\sigma_T^*(\delta)\gamma_T\delta^{T-1}u_P^R}{(\sum_{k=1}^{T-1}\gamma_k+\sigma_T^*(\delta)\gamma_T)}>0
        \end{align*}
        and 
        \begin{align*}
            L_2(\sigma_T^*(\delta),\delta)&=\sum_{t=1}^T\alpha_t^P(\sigma_T^*(\delta))\left((t-1)\delta^{t-2}u_P^R-t\delta^{t-1}U_{t+1}^R(\sigma^*,p^*,\delta)-\delta^t\frac{\partial U_{t+1}^R(\sigma^*,p^*,\delta)}{\partial\delta}\right)\\&<\sum_{t=1}^T\alpha_t^P(\sigma_T^*(\delta))\left((t-1)\delta^{t-2}u_P^R-t\delta^{t-1}U_{t+1}^R(\sigma^*,p^*,\delta)\right)<L(\sigma_T^*(\delta),\delta)=0.
        \end{align*}
        Therefore, by the implicit function theorem,
        \begin{align*}
            \frac{d\sigma_T^*(\delta)}{d\delta}=-\frac{L_2(\sigma_T^*(\delta),\delta)}{L_1(\sigma_T^*(\delta),\delta)}>0
        \end{align*}
        and so $\sigma_T^*(\delta)$ is increasing in $\delta$. Observe that
        \begin{align*}
            Pr(\hat{T}_\delta=\emptyset)=\sigma_T^*(\delta)(1-p^*)^T>\sigma_T^*(\delta')(1-p^*)^T=Pr(\hat{T}_{\delta'}=\emptyset).
        \end{align*}

        Now, observe that for any $t<T$,
        \begin{align*}
            Pr(\hat{T}_\delta\leq t|\hat{T}_\delta\neq\emptyset)&=\sum_{k=1}^t\left(\frac{p^*(1-p^*)^{k-1}}{1-\sigma_T^*(\delta)(1-p^*)^T}\right)\\&<\sum_{k=1}^t\left(\frac{p^*(1-p^*)^{k-1}}{1-\sigma_T^*(\delta')(1-p^*)^T}\right)=Pr(\hat{T}_{\delta'}\leq t|\hat{T}\neq\emptyset)
        \end{align*}
        so the distribution of $\hat{T}_{\delta}$ given $\hat{T}_{\delta}\neq\emptyset$ first-order stochastically dominates the distribution of $\hat{T}_{\delta'}$ given $\hat{T}_{\delta'}\neq\emptyset$. Therefore,
        \begin{align*}
            E[\hat{T}_\delta|\hat{T}_\delta\neq\emptyset]>E[\hat{T}_{\delta'}|\hat{T}_{\delta'}\neq\emptyset].
        \end{align*}
        \item As $T\to\infty$, since $Pr(\hat{T}=\emptyset)=\sigma_T^*(1-p^*)^T<(1-p^*)^T$ and $(1-p^*)^T\to0$ as $T\to\infty$, then $Pr(\hat{T}=\emptyset)\to0$. Therefore,
        \begin{align*}
            E[\hat{T}|\hat{T}\neq\emptyset]=\left(\frac{p^*}{1-Pr(\hat{T}=\emptyset)}\right)\sum_{t=1}^T(1-p^*)^{t-1}t\to p^*\sum_{t=1}^\infty(1-p^*)^{t-1}t=\frac{1}{p^*}
        \end{align*}
        as $T\to\infty$.
    \end{enumerate}
\end{proofof}

\begin{proofof}{\bf Theorem \ref{T:exanteequiv}}
    Let $(\sigma_T^*,p^*)$ be a strategy profile. I consider three exhaustive cases: $\sigma_T^*=1$, $\sigma_T^*=0$, and $\sigma_T^*\in(0,1)$. I show that the conditions required of ex-ante equilibria in these three cases correspond to exactly the three cases in Theorem \ref{T:char}. Observe throughout that (i) in the definition of MBE and ex-ante equilibrium are identical. So I verify that conditions (iii) of MBE and (ii) of ex-ante equilibrium are equivalent in each case.

    \noindent\textbf{Case I $(\sigma_T^*=1)$:} Observe that Equation \eqref{E:exante} can be written as
    \begin{align*}
        \max_p\sum_{t=1}^T\delta^{t-1}(1-p)^{t-1}pu_P^R
    \end{align*}
    which is uniquely solved by $p^*=1$. This is the MBE from Case 1 of Theorem \ref{T:char}.

    \noindent\textbf{Case II $(\sigma_T^*=0)$:} Observe that any $p^*$ solves Equation \eqref{E:exante}. This corresponds to Case 2 of Theorem \ref{T:char}.

    \noindent\textbf{Case III $(\sigma_T^*\in(0,1)$:} Follows immediately from Lemma \ref{L:equiv}.

\end{proofof}

    \setlength{\bibsep}{0pt plus 0.3ex}
\bibliographystyle{ecta}
\bibliography{example}

\newpage
\section*{Supplemental Materials}
\section{Sequential MBE}\label{A:sequential}

Consider the setting from Section \ref{S:analysis}. While the proposer preferred and the mixing equilibrium in Theorem \ref{T:char} do not depend on the respondent's beliefs off the equilibrium path, the respondent preferred equilibrium does. In this section, I verify that the fair equilibrium satisfies a natural perfection refinement. In particular, I apply the notion of a \textit{multiself sequential equilibrium} from \cite{lambertmarpleshoham2019} to my setting.\footnote{I focus on sequential equilibrium in light of Corollary 2 of \cite{lambertmarpleshoham2019}.}

\begin{defn}[Sequential MBE]
    An MBE $(\sigma^*,p^*,\alpha)$ is a \textit{sequential Multiselves Bargaining Equilibrium} if there exists a sequence $(\sigma^{(n)},p^{(n)},\alpha^{(n)})\to(\sigma^*,p_G^*,\alpha)$ of completely mixed strategies\footnote{That is $\sigma_t^{(n)},p^{(n)}\in(0,1)$ for all $t,n$.} such that for all $t,n$
    \begin{align*}
        \alpha_t^{(n)}=\frac{\sigma_t^{(n)}\gamma_t^{(n)}}{\sum_{k=1}^T\sigma_k^{(n)}\gamma_k^{(n)}}
    \end{align*}
    where I define
    \begin{align*}
        \gamma_t^{(n)}=\frac{(1-p^{(n)})^{t-1}\prod_{\ell=1}^{t-1}\sigma_\ell^{(n)}}{1+\sum_{k=1}^{T-1}(1-p)^k\prod_{\ell=1}^k\sigma_\ell^{(n)}}
    \end{align*}
\end{defn}

Clearly, the proposer preferred equilibrium and mixing equilibrium are sequential MBE. The former satisfies the condition since the respondent accepts any off-path offer, regardless of her beliefs. The latter is a sequential MBE since all offers are on-path. The following proposition verifies that the respondent preferred MBE is also sequential.

\begin{prop}
    Every MBE $(\sigma^*,p^*,\alpha)$ is a sequential MBE.
\end{prop}
\begin{proof}
    The proposer equilibrium and the mixing equilibrium from Theorem \ref{T:char} are clearly sequential MBE --- off-path beliefs are irrelevant. I verify, by construction, that the respondent preferred MBE from Theorem \ref{T:char} is a sequential MBE. Let $(\sigma^*,p^*,\alpha)$ be an MBE such that $\sigma_t^*=0$ for all $t$ and $p^*\leq u_R^P(1-\delta)(u_P^P-\delta u_R^P)^{-1}$. Let $\alpha=1$. Then, let $p^{(n)}=p$, let $\sigma_t^{(n)}=n^{-t}$. Then, clearly $\sigma_t^{(n)}\to0=\sigma_t^*$ as $n\to\infty$. Observe next that 
    \begin{align*}
        \gamma_t^{(n)}=\frac{(1-p^*)^{t-1}n^{-\sum_{\ell=1}^{t-1}\ell}}{1+\sum_{k=1}^{T-1}n^{-\sum_{\ell=1}^k\ell}}\to0
    \end{align*}
    and if $t>1$
    \begin{align*}
        \alpha_t^{(n)}=\frac{n^{-t}\gamma_t^{(n)}}{\sum_{k=1}^Tn^{-k}\gamma_k^{(n)}}\to\lim_{n\to\infty}\frac{-tn^{-t-1}\gamma_t^{(n)}+n^{-t}\frac{d\gamma_t^{(n)}}{dn}}{\log(n)\gamma_1^{(n)}+\sum_{k=2}^T-kn^{-k-1}\gamma_k^{(n)}+n^{-k}\frac{d\gamma_k^{(n)}}{dn}}=0
    \end{align*}
    by L'Hopital's rule. If $t=1$,
    \begin{align*}
        \alpha_t^{(n)}=\frac{n^{-t}\gamma_t^{(n)}}{\sum_{k=1}^Tn^{-k}\gamma_k^{(n)}}\to&\lim_{n\to\infty}\frac{\log(n)\gamma_t^{(n)}+n^{-t}\frac{d\gamma_t^{(n)}}{dn}}{\log(n)\gamma_1^{(n)}+\sum_{k=2}^T-kn^{-k-1}\gamma_k^{(n)}+n^{-k}\frac{d\gamma_k^{(n)}}{dn}}=1
    \end{align*}
    so $\alpha^{(n)}\to\alpha$, as desired.
\end{proof}

\section{Absentminded Proposer}\label{S:propforg}

Often, results in bargaining are quite sensitive to the bargaining protocol --- that is, they depend on who makes offers and when. How robust are the delay results from the baseline model to alternative specifications of the bargaining protocol? I'll focus on the case where the proposer is absentminded, but the respondent has perfect recall. For ease of exposition, suppose that $T=2$. The main qualitative conclusions extend simply to the case where $T\geq2$. In this setting, the proposer strategy is $\phi\in[0,1]$, which is the probability of offering $a_P$. The respondent's strategy is a pair $(q_1,q_2)$ where $q_t\in[0,1]$ is the probability of accepting the proposer's preferred offer in period $t$. As before, I assume that the respondent always accepts an offer of $a_R$. The proposer's belief $\alpha$ that the current period is $t=1$ is consistent with $(\phi,q)$ if
\begin{align*}
    \alpha=\gamma=\frac{1}{1+(1-q_1)\phi}.   
\end{align*}
Since the proposer is absentminded and the respondent has perfect recall, offers no longer convey information regarding the calendar time in equilibrium. Thus, the proposer's belief that she is in $t=1$ exactly coincides with the long-run frequency of being in $t=1$. I define MBE in this setting analogously.
\begin{defn}[MBE --- Proposer Absentminded]
    A strategy profile $(\phi^*,q^*)$ and beliefs $\alpha$ is an MBE if 
    \begin{enumerate}[(i)]
        \item $q_t^*$ is a best response to $\phi^*$. That is, $q_2^*=1$ and
        \begin{align*}
           q_1^*\in\arg\max_{q\in[0,1]}qu_P^R+(1-q)\delta\left(\phi^*u_P^R+(1-\phi^*)u_R^R\right) 
        \end{align*}
        \item $\alpha$ is consistent with $(\phi^*,q^*)$.
        \item $\phi^*$ is a best response to $q^*$ given beliefs $\alpha$
        \begin{align*}
            \phi^*\in\arg\max_{\phi\in[0,1]}\;&\alpha\left(\phi q_1^*u_P^P+(1-\phi)u_R^P+\delta\phi(1-q_1^*)\left(\phi^*u_P^P+(1-\phi^*)u_R^P\right)\right)\\&+(1-\alpha)\left(\phi u_P^P+(1-\phi)u_R^P\right)
        \end{align*}
    \end{enumerate}
\end{defn}

As in the case studied in Section \ref{S:analysis}, there is clearly an MBE where $\phi^*=1$ and $q_1^*=1$. The proposer knows that the respondent always accepts her preferred offer, and has no profitable deviation. If the respondent rejects $a_P$ in $t=1$, he receives $a_P$ for certain in period $t=2$. Therefore, the respondent has no profitable deviation.

    If $\delta\geq u_P^R/u_R^R$, then there is an MBE where $\phi^*=0$ and $q_1^*=0$. Given these strategies, the proposer's belief over calendar time is degenerate on $t=1$. Therefore, she chooses $\phi^*$ to maximize
    \begin{align*}
        &\max_\phi\phi q_1^*u_P^P+(1-\phi)u_R^P+\delta\phi(1-q_1^*)u_R^P\\\equiv&\max_\phi(1-\phi)u_R^P+\delta\phi u_R^P
    \end{align*}
    which is solved by $\phi^*=0$. The respondent rationally rejects $a_P$ whenever $u_P^R\leq\delta u_R^R$, or when $\delta\geq u_P^R/u_R^R$.

    As in Section \ref{S:analysis}, my main case of interest is when there is a mixed MBE where $\phi^*,q_1^*\in(0,1)$. In order for $q_1^*\in(0,1)$, it must be the case that
    \begin{align*}
        u_P^R=\delta\left(\phi^*u_P^R+(1-\phi^*)u_R^R\right)\iff \phi^*=\frac{\frac{u_P^R}{\delta}-u_R^R}{u_P^R-u_R^R}
    \end{align*}
    and $q_1^*$ solves
    \begin{align}\label{E:foc}
        \alpha \left(q_1^*u_P^P-u_R^P+\delta(1-q_1^*)\left(\phi^* u_P^P+(1-\phi^*)u_R^P\right)\right)+(1-\alpha)(u_P^P-u_R^P)=0\nonumber\\\iff \left(q_1^*u_P^P-u_R^P+\delta(1-q_1^*)\left(\phi^* u_P^P+(1-\phi^*)u_R^P\right)\right)+(1-q_1^*)\phi^*(u_P^P-u_R^P)=0
    \end{align}
    in order for the proposer to be willing to randomize.

    One can verify that there is a solution $q_1^*\in(0,1)$ to Equation \eqref{E:foc} only for intermediate values of $\delta$. That is, there is a mixed MBE only when $\delta\in[u_P^R/u_R^R,\overline{\delta}]$ for some $1>\overline{\delta}>u_P^R/u_R^R$.\footnote{The proof of Proposition \ref{P:forgetfulchar} gives an explicit expression for $\overline{\delta}$.} If $\delta$ is too high, no mixed MBE exists; the proposer can never be incentivized to offer $a_R$, even if $q_1^*$ is very small. The proposer knows that $a_P$ will always be accepted in $t=2$. Therefore, for high $\delta$, the proposer will always offer $a_P$ even if it is very likely to be rejected in $t=1$.

    \begin{prop}\label{P:forgetfulchar}
        Let $T=2$. There are at most three strategy profiles $(\phi^*,q_1^*)$ that can occur in an MBE when the proposer is absentminded:
        \begin{enumerate}
            \item (Proposer Preferred Equilibrium): $\phi^*=1$ and $q_1^*=1$.
            \item (Respondent Preferred Equilibrium): $\phi^*=0$ and $q_1^*=0$ only if $\delta\geq u_P^R/u_R^R$.
            \item (Mixed MBE): $\phi^*=\frac{\frac{u_P^R}{\delta}-u_R^R}{u_P^R-u_R^R}$ and $q_1^*\in(0,1)$ solves Equation \eqref{E:foc} only if $\delta\in[u_P^R/u_R^R,\overline{\delta}]$.
        \end{enumerate}
    \end{prop}
\begin{proof}
    Let $(\phi^*,q_1^*)$ compose an MBE. I'll consider three exhaustive cases: $\phi^*=1$, $\phi^*=0$, and $\phi^*\in(0,1)$. I show that these correspond to cases 1, 2, and 3 in the statement of Proposition \ref{P:forgetfulchar}.

    \noindent\textbf{Case I $(\phi^*=1)$:} If $\phi^*=1$, the respondent's program in $t=1$ is 
    \begin{align*}
        \arg\max_qqu_P^R+\delta(1-q)u_P^R
    \end{align*}
    which is solved by $q_1^*=1$, as desired. Clearly, $\phi^*=1$ solves the proposer's program.

    \noindent\textbf{Case II $(\phi^*=0)$:} If $\phi^*=0$, the proposer's program in $t=1$ is
    \begin{align*}
        \arg\max_qqu_P^R+\delta(1-q)u_R^R
    \end{align*}
    which is solved by $q_1^*=0$ if and only if $\delta u_R^R\geq u_P^R\iff\delta\geq u_P^R/u_R^R$. The proposer's program is
    \begin{align*}
        &\max_\phi\phi q_1^*u_P^P+(1-\phi)u_R^P+\delta\phi(1-q_1^*)u_R^P\\\equiv&\max_\phi(1-\phi)u_R^P+\delta\phi u_R^P
    \end{align*}
    since $\alpha=1$. So $\phi^*=0$ solves the proposer's program.

    \noindent\textbf{Case III $(\phi^*\in(0,1))$:} If $\phi^*\in(0,1)$, $q_1^*$ must solve Equation \eqref{E:foc}, since the slope of the proposer's program must be $0$. Suppose first that $q_1^*$ which solves Equation \eqref{E:foc} satisfies $q_1^*\in(0,1)$. Then,
    \begin{align*}
        u_P^R=\delta\left(\phi^*u_P^R+(1-\phi^*)u_R^R\right)\iff \phi^*=\frac{\frac{u_P^R}{\delta}-u_R^R}{u_P^R-u_R^R}
    \end{align*}
    which requires $u_P^R/\delta\leq u_R^R\iff \delta\geq u_P^R/u_R^R$. Since Equation \eqref{E:foc} is linear in $q_1^*$, I find a range of $\delta$ for which
    \begin{align*}
        \left(q_1u_P^P-u_P^R+\delta(1-q_1)\left(\phi^*u_P^P+(1-\phi^*)u_R^P\right)\right)+(1-q_1)\phi^*(u_P^P-u_R^P)
    \end{align*}
    is weakly negative when evaluated at $q_1=0$ and weakly positive when evaluated at $q_1=1$. The above expression is positive when evaluated at $q_1=1$ for any $\delta$, since $u_P^P>u_R^P$. When $q_1=0$, we have that
    \begin{align*}
        \delta\left(\phi^*u_P^P+(1-\phi^*)u_R^P\right)+\phi^*(u_P^P-u_R^P)\leq u_R^P\\\iff \frac{u_P^R-\delta u_R^R}{u_P^R-u_R^R}u_P^P+\frac{\delta u_P^R-u_P^R}{u_P^R-u_R^R}u_R^P+\frac{\frac{u_P^R}{\delta}-u_R^R}{u_P^R-u_R^R}(u_P^P-u_R^P)\leq u_R^P
    \end{align*}
    which holds if and only if
    \begin{align*}
        \delta\leq\overline{\delta}=\frac{(u_P^P-2u_R^P)(u_P^R-u_R^R)+\sqrt{(u_P^P-2u_R^P)^2(u_P^R-u_R^R)^2+4(u_R^Ru_P^P-u_R^Pu_P^R)u_R^P(u_P^P-u_R^P)}}{2(u_R^Ru_P^P-u_R^Pu_P^R)}.
    \end{align*}
    The above threshold $\overline{\delta}<1$ since
    \begin{align*}
        \frac{u_P^R-u_R^R}{u_P^R-u_R^R}u_P^P+\frac{u_P^R-u_P^R}{u_P^R-u_R^R}u_R^P+\frac{u_P^R-u_R^R}{u_P^R-u_R^R}(u_P^P-u_R^P)\\=\frac{(u_P^R-u_R^R)(u_P^P+u_P^P)}{u_P^R-u_R^R}=2u_P^P>u_R^P
    \end{align*}
    and the above threshold $\overline{\delta}>u_P^R/u_R^R$ since
    \begin{align*}
        \frac{u_P^R-u_P^R}{u_P^R-u_R^R}u_P^P&+\frac{(u_P^R)^2/u_R^R-u_P^R}{u_P^R-u_R^R}u_R^P+\frac{u_R^R-u_R^R}{u_P^R-u_R^R}(u_P^P-u_R^P)\\&=\frac{u_P^R(u_P^R/u_R^R-1)}{u_P^R-u_R^R}u_R^P<\frac{u_R^R(u_P^R/u_R^P-1)}{u_P^R-u_R^R}u_R^P\\&=u_R^P
    \end{align*}
    as desired.
\end{proof}

    The key difference between the mixed MBE in Proposition \ref{P:forgetfulchar} and the mixed MBE studied in Section \ref{S:analysis} is that, when the respondent has perfect recall, there is no risk of a trade falling through. There is still delay, however, in the mixed strategy MBE. Delay occurs with probability $\phi^*(1-q_1^*)$. As $\delta\to\overline{\delta}$, $q_1^*\to0$. Therefore, the probability of delay approaches $(\frac{u_P^R}{\overline{\delta}}-u_R^R)({u_P^R-u_R^R})^{-1}$. There is a sharp discontinuity in the probability of delay at $\bar{\delta}$, since delay cannot occur in any MBE when $\delta>\bar{\delta}$.

\section{Two Absentminded Players}\label{S:bothforg}

Another natural extension is the case where both parties are absentminded. Consider the non-transferable utility setting. Unlike the baseline model, the offer does not contain information about the calendar time. For simplicity, assume that $T=2$. The main qualitative conclusions extend readily to the case where $T\geq2$.

A proposer strategy is a value $\sigma^*\in[0,1]$ representing the probability of offering $a_P$ and a respondent strategy is a value $p^*\in[0,1]$ which is the probability of accepting $a_P$. Fair offers are always accepted, as in the baseline model. The respondent holds a belief $\alpha\in[0,1]$ that the calendar time is $t=1$. If $\sigma^*>0$, then $\alpha$ is consistent if
\begin{align*}
    \alpha=\frac{\sigma^*\gamma}{\sigma^*\gamma+\sigma^*(1-\gamma)}=\gamma
\end{align*}
where
\begin{align*}
    \gamma=\frac{1}{1+(1-p^*)\sigma^*}.
\end{align*}
Hence, offering $a_P$ contains no information about calendar time.

\begin{defn}[MBE --- Two Forgetful Players]
    A strategy profile $(\sigma^*,p^*)$ and beliefs $\alpha$ is an MBE if
    \begin{enumerate}[(i)]
        \item $\sigma^*$ is a best response given $p^*$. That is, 
        \begin{align*}
            \sigma^*\in\arg\max_{s\in[0,1]}sp^*u_P^P+(1-s)u_R^P+s(1-p^*)\alpha\delta\left(\sigma^*p^*u_P^P+(1-\sigma^*)u_R^P\right)
        \end{align*}
        \item $p^*$ is a best response to $\sigma^*$. That is,
        \begin{align*}
            p^*\in\arg\max_{p\in[0,1]}pu_P^R+(1-p)\alpha\delta\left(\sigma^*p^*u_P^R+(1-\sigma^*)u_R^R\right)
        \end{align*}
        \item $\sigma^*>0$ implies that $\alpha=\gamma$.
    \end{enumerate}
\end{defn}

As in the baseline model, it is immediate that there is an MBE with $\sigma^*=1$ and $p^*=1$. There also is an MBE with $\sigma^*=0$ and $p^*=0$ for $\delta\geq u_P^R/u_R^R$. Moreover, there is an equilibrium with delay where both players mix with probabilities $\sigma^*\in(0,1)$ and $p^*\in(0,1)$ satisfying
\begin{align}\label{E:two1}
        p^*u_P^P+(1-p^*)\gamma\delta\left(\sigma^*p^*u_P^P+(1-\sigma^*)u_R^P\right)=u_R^P
    \end{align}
and
\begin{align}\label{E:twotwo}
        u_P^R=\gamma\delta\left(\sigma^*p^*u_P^R+(1-\sigma^*)u_R^R\right).
    \end{align}
    Similar to the baseline model, an interior solution to Equations \eqref{E:two1} and \eqref{E:twotwo} exists and is unique if and only if $\delta\geq u_P^R/u_R^R$.
\begin{prop}\label{P:twoabsent}
    Let $T=2$ and $\delta\geq u_P^R/u_R^R$. The following cases exhaust all possibilities of strategy profiles $(\sigma^*,p^*)$ that can occur in MBE when both parties are absentminded:
    \begin{enumerate}
        \item (Proposer Preferred Equilibrium): $\sigma^*=1$ and $p^*=1$.
        \item (Respondent Preferred Equilibrium): $\sigma^*=0$ and $p^*\leq\frac{u_R^P(1-\delta)}{u_P^P-\delta u_R^P}$.
        \item (Mixing Equilibrium): $\sigma^*\in(0,1)$ and $p^*\in(0,1)$ solve Equations \eqref{E:two1} and \eqref{E:twotwo}.
    \end{enumerate}
\end{prop}
\begin{proof}
    Let $(\sigma^*,p^*)$ compose an MBE. I'll consider three exhaustive cases: $\sigma^*=1$, $\sigma^*=0$, and $\sigma^*\in(0,1)$. I'll show that the first two correspond to cases 1, 2, and 3 of Proposition \ref{P:twoabsent}.

    \noindent\textbf{Case I $(\sigma^*=1)$:} If $\sigma^*=1$, the respondent's program is
    \begin{align*}
        \arg\max_{p\in[0,1]}pu_P^R+(1-p)\gamma\delta p^*u_P^R.
    \end{align*}
    Since $u_P^R>\gamma\delta p^*u_P^R$, then $p^*=1$ solves the respondent's program.

    \noindent\textbf{Case II $(\sigma^*=0)$:} If $\sigma^*=0$, observe that $\gamma=1$. Since the proposer's program is
    \begin{align*}
        \arg\max_{s\in[0,1]}sp^*u_P^P+(1-s)u_R^P+s(1-p^*)\delta u_R^P
    \end{align*}
    which is solved by $\sigma^*=0$ if and only if
    \begin{align*}
        p^*u_P^P+(1-p^*)\delta u_R^P&\leq u_R^P\\\iff p^*\leq\frac{u_R^P(1-\delta)}{u_P^P-\delta u_R^P}.
    \end{align*}
    Next, the respondent's program is
    \begin{align*}
        \arg\max_{p\in[0,1]}pu_P^R+(1-p)\alpha\delta u_R^R.
    \end{align*}
    If $p^*=0$, $\alpha=1$ supports $(\sigma^*,p^*)$ as an MBE. If $p^*>0$, $\alpha$ solving
    \begin{align*}
        u_P^R&=\alpha\delta u_R^R\\\iff\alpha&=\frac{u_P^R}{\delta u_R^R}
    \end{align*}
    supports $(\sigma^*,p^*)$ as an MBE.

    \noindent\textbf{Case III $(\sigma^*\in(0,1)$:} Let $\sigma^*\in(0,1)$. Since $\sigma^*>0$, $\alpha=\gamma$. The proposer's program satisfies
    \begin{align*}
        p^*u_P^P+(1-p^*)\gamma\delta\left(\sigma^*p^*u_P^P+(1-\sigma^*)u_R^P\right)=u_R^P
    \end{align*}
    which does not hold if $p^*=0$ or if $p^*=1$. So $p^*\in(0,1)$. Therefore, the respondent's program satisfies
    \begin{align*}
        u_P^R=\gamma\delta\left(\sigma^*p^*u_P^R+(1-\sigma^*)u_R^R\right).
    \end{align*}
    These are equations \eqref{E:two1} and \eqref{E:twotwo}, respectively. Define the following expressions:
    \begin{align*}
        Q_{1,\sigma}(p)&=pu_P^P+\frac{\delta(1-p)(\sigma pu_P^P+(1-\sigma)u_R^P)}{1+(1-p)\sigma}-u_R^P\\Q_{2,p}(\sigma)&=\frac{\delta(\sigma p u_P^R+(1-\sigma)u_R^R)}{1+(1-p)\sigma}-u_P^R.
    \end{align*}
    Equations \eqref{E:two1} and \eqref{E:twotwo} have a solution $(\sigma^*,p^*)\in(0,1)^2$ if and only if $Q_{1,\sigma^*}(p^*)=0$ and $Q_{2,p^*}(\sigma^*)=0$. Observe first that, given $\sigma\in[0,1]$, $Q_{1,\sigma}(p)$ is a concave in $p$ with
    \begin{align*}
        Q_{1,\sigma}(0)&=\frac{\delta(1-\sigma)u_R^P}{1+\sigma}-u_R^P<0\\
        Q_{1,\sigma}(1)&=u_P^P-u_R^P>0
    \end{align*}
    so $Q_{1,\sigma}(p)=0$ has a unique solution in $[0,1]$ for any $\sigma\in[0,1]$. Moreover, observe that
    \begin{align*}
        Q_{2,p}'(\sigma)=\frac{(1+(1-p)\sigma)(\delta pu_P^R-u_R^R)-\delta(\sigma pu_P^R+(1-\sigma)u_R^R)(1-p)}{(1+(1-p)\sigma)^2}<0.
    \end{align*}
    And observe that
    \begin{align*}
        Q_{2,p}(0)&=\delta u_R^R-u_P^R\geq0\\Q_{2,p}(1)&=\frac{\delta pu_P^R}{2-p}-u_P^R<0.
    \end{align*}
    Since $Q_{2,p}(\sigma)$ is monotonically decreasing in $\sigma$, positive at $\sigma=0$, and negative at $\sigma=1$, $Q_{2,p}(\sigma)=0$ has a unique solution in $[0,1]$ for any $p\in[0,1]$.
    
\end{proof}

\section{A Model with Side Payments}\label{S:general}

In many settings where parties bargain over a course of action, side payments are available. For instance, when firms negotiate over which technology to adopt during a merger, the firm that must adapt can be compensated for the switching cost. As it turns out, access to transfers considerably changes the analysis. I'll present a model with a \textit{single} action and side payments, to highlight the role played by transfers. This model amounts to bargaining over how to divide the surplus from taking the action. The key qualitative insights extend readily to the case with multiple actions.

Formally, two players bargain over a fixed pie of size $V>0$, representing the total undiscounted surplus from taking the action. Suppose, for ease of exposition, that the trade deadline is $T=2$.\footnote{Proposition \ref{P:cont} extends easily to the case where $T\geq2$.} One player is the proposer, who offers a share $x\in[0,1]$ of the pie to the respondent in each period where no previous offer has been accepted. Players discount the future, so the proposer's payoff from reaching a deal when offer $x$ is made in period $t$ is $\delta^{t-1}(1-x)V$ and the respondent's payoff is $\delta^{t-1}xV$.

A history in period $t$ is a sequence $h_t=(x_1,...,x_{t-1})$ of past (rejected) offers. Denote the set of histories by $\mathcal{H}_t$. Let $\sigma_t:\mathcal{H}_t\to\Delta [0,1]$ denote the proposer's behavioral strategy in period $t$ and let $p:[0,1]\to[0,1]$ map an offer $x$ to the probability $p(x)$ that the respondent accepts an offer. I'll often abuse notation in the following manner: I'll let $\sigma_1(x)$ denote the offer distribution in period $t=1$ and let $\sigma_2(\hat{x})[x_1]$ denote the offer distribution at $t=2$ given that $x_1$ was offered in $t=1$. Let $\alpha(x)$ denote the respondent's belief that the calendar time is $t=1$ after receiving an offer $x$. The belief rule $\alpha(x)$ is consistent with $(\sigma_1,\sigma_2)$ and $p(x)$ if $\alpha(x)$ is formed via Bayes' rule wherever possible. Observe also that the respondent's unconditional belief that the time is $t=1$ is
\begin{align*}
    \gamma=\frac{1}{1+\int_0^1(1-p(x))d\sigma_1(x)}.
\end{align*}

A Multiselves Bargaining Equilibrium can be defined in a manner similar to the baseline model. Continuation values are simpler to express since $T=2$.

\begin{defn}[MBE --- With Side Payments]
    A proposer strategy $(\sigma_1^*,\sigma_2^*)$, a respondent strategy $p^*(x)$, and beliefs $\alpha(x)$ is an Multiselves Bargaining Equilibrium (MBE) if
    \begin{enumerate}[(i)]
        \item $(\sigma_1^*,\sigma_2^*)$ is sequentially rational given $p^*$. That is, if $\tilde{x}_1\in\supp\sigma_1^*$
        \begin{align*}
            \tilde{x}_1\in\arg\max_x\;(1-x)Vp^*(x)+\delta(1-p^*(x))\int_0^1(1-\hat{x})Vp^*(\hat{x})d\sigma_2^*(\hat{x})[\tilde{x}_1]
        \end{align*}
        and if $\tilde{x}_2\in\supp\sigma_2^*[x_1]$ for some $x_1\in[0,1]$, 
        \begin{align*}
            \tilde{x}_2\in\arg\max_x\;(1-x)Vp^*(x)
        \end{align*}
        \item $\alpha(x)$ is consistent with $(\sigma_1^*,\sigma_2^*)$ and $p^*(x)$.
        \item For each $x$, $p^*(x)$ solves
        \begin{align*}
            \max_p\left\{\alpha(x)\left(pxV+(1-p)\delta\int_0^1\hat{x}Vp^*(\hat{x})d\sigma_2^*(\hat{x})[x]\right)+(1-\alpha(x))\delta pxV\right\}
        \end{align*}
        \item If $x\in\supp\sigma_1^*$, $p^*(x)>0$.
    \end{enumerate}
\end{defn}

The restriction to Markov strategies for the proposer is not without loss of generality when side payments are possible. In particular, the set of possible outcomes drastically shrinks if one requires strategies to satisfy a Markov property, as I will show in Propositions \ref{P:cont} and \ref{P:contdelay}.
\begin{defn}[Markov Property]
    An MBE $(\sigma_1^*,\sigma_2^*)$, $p^*(x)$, and $\alpha(x)$ satisfies the Markov property if $\sigma_2^*[x_1]=\sigma_2^*[x_1
']$ for all $x_1,x_1'\in [0,1]$. Often, I call MBE which satisfy the Markov property "Markov Perfect MBE".\footnote{See, e.g., \cite{maskintirole2001} for a definition of Markov Perfect equilibria in dynamic games with perfect recall.}
\end{defn}

Finally, an MBE has delay if there is positive probability of reaching period $t=2$. That is, there is delay whenever there is no immediate agreement. 
\begin{defn}[Delay]
    An MBE $(\sigma_1^*,\sigma_2^*)$, $p^*(x)$, $\alpha(x)$ has delay if 
    \begin{align*}
        \int_Xp^*(x)d\sigma_1^*(x)<1.
    \end{align*}
\end{defn}

\begin{prop}\label{P:acceptany}
    If $\delta<1$, in any Markov Perfect MBE, $\sigma_2^*$ puts probability $1$ on $x=0$. If $\delta=1$, $\sigma_2^*$ puts probability $1$ on some offer $x'\in[0,1]$.
\end{prop}
\begin{proof}
    Observe first that, if $\tilde{x}_1>\delta\int_0^1\hat{x}p^*(\hat{x})d\sigma_2^*(\hat{x})$, then $p^*(\tilde{x}_1)=1$ by condition (i) of the definition of MBE. Therefore, since
    \begin{align*}
        (1-x)p^*(x)\geq1-\delta\int_0^1\hat{x}p^*(\hat{x})d\sigma_2^*(\hat{x})
    \end{align*}
    it must be the case that for any $\tilde{x}_2\in\supp\sigma_2^*$, $\tilde{x}_2\leq\delta\int_0^1\hat{x}p^*(\hat{x})d\sigma_2^*(\hat{x})$. Therefore,
    \begin{align*}
        \supp\sigma_2^*\subseteq\left[0,\delta\int_0^1\hat{x}d\sigma_2^*(\hat{x})\right]
    \end{align*}
    which implies that $\sigma_2^*$ is degenerate on $0$ if $\delta<1$. If $\delta=1$, then $\tilde{x}_2\in\supp\sigma_2^*\implies\tilde{x}_2\leq E_{\sigma_2^*}[x_2]$, so $\tilde{x}_2=E_{\sigma_2^*}[x_2]$, implying that $\sigma_2^*$ is degenerate on some point.
\end{proof}

The following proposition follows immediately from Proposition \ref{P:acceptany}.

\begin{prop}[TU Markov MBE]\label{P:cont}
    The following cases characterize Markov Perfect MBE.
    
    \begin{enumerate}
        \item If $(\sigma_1^*,\sigma_2^*)$ and $p^*(x)$ occur in a Markov Perfect MBE and $\delta<1$, then $\sigma_1^*$ and $\sigma_2^*$ place probability $1$ on $x=0$ and $p^*(x)=1$ for all $x\in[0,1]$.
        \item If $(\sigma_1^*,\sigma_2^*)$ and $p^*(x)$ occur in a Markov Perfect MBE, $\delta=1$, then $\sigma_1^*$ and $\sigma_2^*$ place probability $1$ on some $x'\in[0,1]$ and $p^*(x)=1$ for all $x\in[x',1]$.
    \end{enumerate}
\end{prop}
\begin{proof}
    Suppose $(\sigma_1^*,\sigma_2^*)$ and $p_G^*$ occur in a Markov Perfect MBE. If $\delta<1$, then $\sigma_2^*$ puts probability $1$ on $x=0$. Suppose that for some $x\in[0,1]$, $p^*(x)<1$. Then, if beliefs are given by $\alpha(x)$
    \begin{align*}
        \alpha(x)xV-\alpha(x)\delta\int_0^1\hat{x}Vp^*(\hat{x})d\sigma_2^*(\hat{x})+(1-\alpha(x))\delta xV&\leq0\\\implies \alpha(x)x+(1-\alpha(x))\delta x&\leq 0
    \end{align*}
    which can only hold if $x=0$. So $p^*(x)=1$ for all $x>0$. The only pure strategy equilibrium requires that $p^*(0)=1$ as well. Therefore, $\sigma_1^*$ places probability $1$ on $x=0$.

    If $\delta=1$, there is some $x'\in[0,1]$ such that $\sigma_2^*$ is degenerate on $x'$. Since $\delta=1$, $p^*(x)=1$ if $x\geq x'$, and $p^*(x)=0$ otherwise.
\end{proof}

Proposition \ref{P:cont} shows that the respondent possesses no bargaining power when utility is transferable and players are impatient. This is due to the ability for the proposer to \textit{fine-tune} her offers. The offer space is so rich that in any strategy profile where the proposer makes an acceptable offer $x$ which leaves the respondent with positive surplus in period $t=2$, there is a profitable deviation --- namely, offering $\delta x$. Since the respondent is always willing to accept $\delta x$, there cannot be an equilibrium which leaves the respondent with any surplus. Essentially, since concessions can be arbitrarily small, there is no credible concession made on the equilibrium path that gives the respondent a strictly positive share of the surplus.

When players are patient (i.e. $\delta=1$),\footnote{The game form is well defined even when $\delta=1$ since I assume a finite time horizon.} there are non-Markovian MBE that feature delay. To show this, I'll construct an example of such an MBE. 

\begin{ex}[MBE with Delay]\label{ex:delay}
    Suppose that $\supp\sigma_1^*=\{1/4\}$. That is, the proposer sends a "greedy offer" in $t=1$, which would secure three-fourths of the pie for herself if accepted. On the equilibrium path, the proposer will randomize between the "greedy" offer and a "fair" offer which gives each player half the pie. That is, $\supp\sigma_2^*[1/4]=\{1/4,1/2\}$. Off-path, suppose that $\supp\sigma_2^*[x]=\{1/2\}$ whenever $x\neq1/4$. Hence, the proposer's strategy is non-Markovian.
    
    The respondent strategy is the following: $p^*(1/4)=2/3$ (to hold the proposer indifferent between offering $1/4$ and $1/2$), $p^*(x)=1$ for any $x\geq1/2$, and $p^*(x)=0$ for any $x<1/2$ such that $x\neq1/4$. By an argument similar to Theorem \ref{T:char}, $p^*(1/4)$ is a best response to the proposer's strategy (provided the proposer randomizes in $t=2$ in such a manner that the respondent is indifferent between accepting and rejecting). Since $x=1/2$ is the best offer the respondent can receive in equilibrium, clearly $p^*(x)=1$ for any $x\geq1/2$. Finally, the respondent should optimally reject any off-path offer below $1/2$. Setting the respondent's belief to $\alpha^*(x)=1$ after $x<1/2$ with $x\neq1/4$ supports this equilibrium, since receiving an off-path offer below $1/2$ leads the respondent to conclude that he is in period $t=1$ and will receive an offer of $1/2$ for certain in the next period.
\end{ex}

Non-Markovian strategies, and the fact that both parties are patient, allow one to effectively replicate the behavior exhibited in the baseline model. This is because the respondent rationally rejects any offer below $1/2$ (except for $1/4$), effectively constraining the space of offers to a binary set $\{1/4,1/2\}$. The respondent rejects the offer $x=1/4$ with positive probability, in hopes that $t=1$ and he will receive a fair offer of $x=1/2$ in $t=2$. The proposer mixes between $x=1/4$ and $x=1/2$ in the second period. Offering $x=1/2$ prevents the deal from falling through, whereas offering $x=1/4$ is rejected with positive probability but secures a larger share of the surplus for the proposer if accepted. 

In the constructed MBE with delay, the respondent is surprised by an offer $x\neq1/4,1/2$. He then becomes \textit{optimistic} that he will receive a more favorable offer in the subsequent period. This mechanism resembles \cite{friedenberg2019}, despite the fact that there is no strategic uncertainty off the equilibrium path.

\begin{prop}[TU MBE with Delay]\label{P:contdelay}
    There exists an MBE with delay if and only if $\delta=1$. Moreover, if $\delta=1$, for any pair $x_L,x_H\in X$ with $x_L<x_H$, there exists an MBE where the proposer offers $x_L$ in period $1$, $p^*(x_L)<1$, and on the equilibrium path the proposer mixes between $x_L$ and $x_H$ in period $2$.
\end{prop}
\begin{proof}
     $(\impliedby):$ Suppose that $\delta=1$. For any points $x_L,x_H\in X$ with $x_L<x_H$, I'll construct an MBE where $\sigma_1^*$ is degenerate on $x_L$ and $\sigma_2^*[x_L]$ mixes between $x_L$ and $x_H$. On the conjectured equilibrium path, $x_H$ is the largest offer that the respondent can receive. So $p^*(x_H)=1$. Clearly, $\forall x\geq x_H$, $p^*(x)=1$. The proposer is willing to mix between $x_L$ and $x_H$ in $t=2$ if
    \begin{align*}
        p^*(x_L)(1-x_L)V&=(1-x_H)V\\p^*(x_L)&=\frac{1-x_H}{1-x_L}.
    \end{align*}
    Sequential rationality for the proposer in $t=2$ requires that $\forall x\not\in\{x_L,x_H\}$, 
    \begin{align*}
        p^*(x)(1-x)V\leq (1-x_H)V
    \end{align*}
    which holds for all $x>x_H$. For $x<x_H$ and $x\neq x_L$, set $p^*(x)=0$. Let $\sigma_2^*[x]$ be degenerate on $x_H$ for all $x\neq x_L$. Thus, the proposer's strategy in $t=2$ is sequentially rational. 
    
    The proposer's strategy in $t=1$ is also sequentially rational. The payoff from offering $x_L$ in $t=1$ to the proposer is
    \begin{align*}
        &p^*(x_L)(1-x_L)V+(1-p^*(x_L))(p^*(x_L)(1-x_L)V+(1-x_H)V)\\&=p^*(x_L)(1-x_L)V+(1-p^*(x_L))(2p^*(x_L)(1-x_L)V)\\&=p^*(x_L)(1-x_L)V(1+2(1-p^*(x_L))>(1-x_H)V
    \end{align*}
    so $x_L$ is strictly a best response in $t=1$.

    Finally, $p^*(x)$ is sequentially rational if $x\geq x_H$. If $x<x_H$ and $x\neq x_L$, since $x$ is off-path, let $\alpha(x)=1$. Then rejecting an offer of $x$ is strictly optimal, since rejecting yields a payoff of $x_HV$, whereas accepting yields a payoff of $xV$. Next, observe that
    \begin{align*}
        \alpha(x_L)=\frac{\gamma}{\gamma+(1-q)(1-\gamma)}
    \end{align*}
    where $q\in[0,1]$ is the probability with which the proposer offers $x_H$ in period $t=2$. In order for the respondent to be willing to mix after $x=x_L$, it must be that
    \begin{align*}
        x_LV&=\alpha(x_L)\left(qx_HV+(1-q)x_LV\right)\\&=\frac{\gamma}{\gamma+(1-q)(1-\gamma)}\left(qx_HV+(1-q)x_LV\right)
    \end{align*}
    I'll demonstrate that this has a solution in $q\in(0,1)$ by applying the Intermediate Value Theorem. Observe first that when $q=0$,
    \begin{align*}
        x_LV>\frac{\gamma}{\gamma+(1-q)(1-\gamma)}\left(qx_HV+(1-q)x_LV\right)=\gamma x_LV.
    \end{align*}
    Moreover, when $q=1$,
    \begin{align*}
        x_LV<x_HV
    \end{align*}
    Therefore, there is a solution $q\in(0,1)$. Therefore, this is an MBE. It features delay, since $p^*(x_L)<1$. 

    \noindent$(\implies):$ Suppose that there is an MBE with delay and, for sake of contradiction, suppose that $\delta<1$. Consider the set
    \begin{align*}
        \bar{X}&=\{x\in[0,1]: x\in\supp\sigma_1^*\;\text{or}\;x\in\supp\sigma_2^*[x']\;\text{for some $x'\in\supp\sigma_1^*$}\}
    \end{align*}
    and let $\bar{x}=\sup\bar{X}$. Then $\bar{x}$ is the largest on-path offer. I'll show that $\bar{x}=0$. Suppose not. Since $p^*(x)$ is sequentially rational, it must be the case that $p^*(x)=1$ for all $x\geq\bar{x}$. By sequential rationality, it must be the case that for any $x<\bar{x}$,
    \begin{align*}
        p^*(x)(1-x)V\leq(1-\bar{x})V
    \end{align*}
    and so $p^*(x)<1$ for all $x<\bar{x}$.

    I'll next show that, for each $\epsilon>0$, $\exists x'_\epsilon\in\supp\sigma_1^*$ such that $\max\supp\sigma_2^*[x_\epsilon']+\epsilon\geq\bar{x}$. Suppose not. Then $\bar{x}\in\supp\sigma_1^*$, but $\bar{x}\not\in\supp\sigma_2^*[x']$ for any $x'\in\supp\sigma_1^*$. Then $\forall x<\bar{x}$,
    \begin{align*}
        (1-\bar{x})V&\geq p^*(x)(1-x)V+\delta(1-p^*(x))E_{\sigma_2^*[x]}[(1-x_2)V]\\\iff (1-\bar{x})&> p^*(x)(1-x)+(1-p^*(x))(1-\bar{x})\\\iff p^*(x)(1-\bar{x})&>p^*(x)(1-x)
    \end{align*}
    contradicting the assumption that $\bar{x}>x$. Therefore, without any loss of generality, suppose $\exists x_1'\in\supp\sigma_1^*$ such that $\bar{x}\in\supp\sigma_2^*[x_1']$. So $\bar{x}$ is on-path in $t=2$. 

    We reach a contradiction by observing the following: $p^*(x')=1$ for some $x'\in(\delta\bar{x},\bar{x})$, contradicting the conclusion that $p^*(x)<1$ for all $x<\bar{x}$. Let $x'\in(\delta\bar{x},\bar{x})$. If $p^*(x')<1$, then
    \begin{align*}
        x'V&\leq \alpha(x')\delta E_{\sigma_2^*}[x_2V]\\&\leq \delta E_{\sigma_2^*}[x_2V]<\delta\bar{x}V
    \end{align*}
    contradicting the assumption that $x'>\delta\bar{x}$. Therefore, $\bar{x}=0$. It follows that the only MBE outcome when $\delta<1$ is the Markov Perfect MBE outcome described in Proposition \ref{P:cont}.
\end{proof}

Thus, by Propositions \ref{P:cont} and \ref{P:contdelay}, in the patient case there are equilibria with delay and these equilibria must fail the Markov property. Moreover, a folk theorem emerges: for any pair $x_L<x_H$, there is an MBE with delay where the proposer offers $x_L$ in $t=1$ and mixes between $x_L$ and $x_H$ in $t=2$. Unlike the baseline model, there is no efficiency loss due to delay since $\delta=1$. However, there is strictly positive probability that no deal is reached by the trade deadline $T$. Thus, these MBE are inefficient relative to those without delay. This is analogous to a number of results in the one-sided asymmetric information bargaining literature. In \cite{flt} and \cite{gsw}, the Coase conjecture holds when players use weakly stationary strategies. However, there are equilibria with delay where the parties use non-Markovian strategies \citep{ausubeldeneck}.\footnote{These equilibria only exist in the no-gap case, where the buyer's value distribution is not bounded away from the seller's marginal cost.} 

I now turn to a characterization of all equilibria with delay. First, in any non-trivial equilibrium with delay, one can show that for all $x_1$ in the support of $\sigma_1^*$, $p^*(x_1)\in(0,1)$. Therefore, for all $x_1\in\supp\sigma_1^*$,
\begin{align*}
    x_1=\alpha(x_1)E_{\sigma_2^*[x_1]}[\hat{x}p^*(\hat{x})]\leq E_{\sigma_2^*[x_1]}[\hat{x}p^*(\hat{x})].
\end{align*}
That is, the proposer makes increasingly generous offers (on average) as the trade deadline looms. Moreover, one can show that the lowest $\bar{x}$ which is accepted for certain (that is, $\bar{x}=\inf\{x:p^*(x)=1\}$) is offered on the equilibrium path in $t=2$, and is the \textit{only} offer which reveals the period to be $t=2$. That is, the proposer pools her offers in period $1$ and $2$, with the exception that the maximal on-path offer $\bar{x}$ can be used to reveal that the trade deadline looms. 

\begin{defn}[On-Path Markov Property]
    An MBE $(\sigma_1^*,\sigma_2^*)$ and $p^*$ satisfies the on-path Markov property if $\forall x_1,x_1'\in\supp\sigma_1^*$, $\sigma_2^*[x_1]=\sigma_2^*[x_1']$.
\end{defn}

When restricting attention to equilibria that satisfy an \textit{on-path} Markov requirement (that is, on the equilibrium path, the proposer's period $2$ strategy is independent of her offer in $t=1$), one can show that any equilibrium with delay has a simple binary structure, as in Example \ref{ex:delay}. 

\begin{prop}[Characterizing Equilibria with Delay]\label{P:necessity}
    Let $\delta=1$. Let $(\sigma_1^*,\sigma_2^*)$ and $p^*$ compose an MBE with delay where $p^*(x_1)>0$ for all $x_1\in\supp\sigma_1^*$.
    \begin{enumerate}
        \item For all $x_1\in\supp\sigma_1^*$, $x_1\leq E_{\sigma_2^*[x_1]}[\hat{x}p^*(\hat{x})]$.
        \item For all $x_1\in\supp\sigma_1^*$, $\supp\sigma_2^*[x_1]\subseteq\supp\sigma_1^*\cup\{\bar{x}\}$.
        \item If $(\sigma_1^*,\sigma_2^*)$ and $p^*$ satisfy the on-path Markov property, then $\exists x_L,x_H\in [0,1]$ with $x_L<x_H$ such that $\supp\sigma_1^*=\{x_L\}$ and $\supp\sigma_2^*[x_L]=\{x_L,x_H\}$.
    \end{enumerate}
\end{prop}
\begin{proof}
Throughout, fix an MBE $(\sigma_1^*,\sigma_2^*)$ and $p^*$ with delay where $p^*(x_1)>0$ for all $x_1\in\supp\sigma_1^*$. Since there is delay, $\exists x_1'\in\supp\sigma_1^*$ such that $p^*(x_1')<1$. Let $\bar{x}=\inf\{x:p^*(x)=1\}$. I'll first prove a few intermediate claims.

    \begin{claim}\label{CL1}
        If $x_1\in\supp\sigma_1^*$, then $p^*(x_1)<1$.
    \end{claim}
    \begin{proof}
        Suppose for sake of contradiction that for some $x_1\in\supp\sigma_1^*$, $p^*(x_1)=1$. Then $x_1\geq\bar{x}$. Consider a value $x_1'\in\supp\sigma_1^*$ with $p^*(x_1')<1$. Then, by the proposer's optimality condition, the proposer is indifferent between $x_1$ and $x_1'$
    \begin{align*}
        (1-x_1)p^*(x_1)+(1-p^*(x_1))E_{\sigma_2^*[x_1]}[\hat{x}p^*(\hat{x})]&=(1-x_1')p^*(x_1')+(1-p^*(x_1'))E_{\sigma_2^*[x_1']}[\hat{x}p^*(\hat{x})]
        \\\implies(1-x_1)&\geq (1-x_1')p^*(x_1')+(1-p^*(x_1'))(1-\bar{x})\\\implies (1-x_1)&>(1-\bar{x})
    \end{align*}
    since $x_1'<\bar{x}$. So $\bar{x}> x_1$, a contradiction.
    \end{proof}
    \begin{claim}\label{CL2}
        For some $x_1\in\supp\sigma_1^*$, there exists a sequence $(x_n)\subseteq\bigcup_{x_1\in\supp\sigma_1^*}\supp\sigma_2^*[x_1]$ such that $x_n\to\bar{x}$ as $n\to\infty$.
    \end{claim}
    \begin{proof}
        Suppose for sake of contradiction that , let $\tilde{x}=\sup\{\supp\sigma_2^*[x_1]:x_1\in\supp\sigma_1^*\}$. Clearly, $\tilde{x}<\bar{x}$ and $p^*(\tilde{x})=1$, a contradiction.
    \end{proof}
    
    \begin{claim}
        If $(\sigma_1^*,\sigma_2^*)$ and $p^*$ satisfies the on-path Markov property, then $\bar{x}\in\supp\sigma_2^*[x_1]$ for all $x_1\in\supp\sigma_1^*$.
    \end{claim}
    \begin{proof}
        Denote the on-path offer distribution in $t=2$ by $\tilde{\sigma}_2^*$. Suppose for sake of contradiction that $\bar{x}\not\in\supp\tilde{\sigma}_2^*$. Let $\bar{x}_2=\max\supp\tilde{\sigma}_2^*$. Then $p^*(\bar{x}_2)=1$. Since $\supp\tilde{\sigma}_2^*\subseteq[0,\bar{x}]$, $\bar{x}_2<\bar{x}$, a contradiction. 
    \end{proof}
    \begin{claim}
        If $(\sigma_1^*,\sigma_2^*)$ and $p^*$ satisfies the on-path Markov property. Then there does not exist two distinct $x_1,x_1'\in\supp\sigma_1^*$ with $\alpha(x_1)=\alpha(x_1')=1$.
    \end{claim}
    \begin{proof}
         Denote the on-path offer distribution in $t=2$ by $\tilde{\sigma}_2^*$. Suppose that there exist distinct $x_1,x_1'\in\supp\sigma_1^*$ such that $\alpha(x_1)=\alpha(x_1')$. Since $p^*(x_1),p^*(x_1')<1$, then
        \begin{align*}
            x_1=E_{\tilde{\sigma}_2^*}[\hat{x}p^*(\hat{x})]=x_1'
        \end{align*}
        which is a contradiction.
    \end{proof}
     \begin{claim}
        If $(\sigma_1^*,\sigma_2^*)$ and $p^*$ satisfies the on-path Markov property, then $\exists x_L\in[0,1)$ such that $\supp\sigma_1^*=\{x_L\}$.
    \end{claim}
    \begin{proof}
        Suppose, for sake of contradiction there exist $x_1,x_1'\in\supp\sigma_1^*$ with $x_1<x_1'$. By claim 4, we cannot have $\alpha(x_1)=\alpha(x_1')=1$. Therefore, either $x_1\in\supp\tilde{\sigma}_2^*$ or $x_1'\in\supp\tilde{\sigma}_2^*$. If $x_1\not\in\supp\tilde{\sigma}_2^*$, since $x_1'\in\supp\tilde{\sigma}_2^*$
        \begin{align*}
            x_1&=E_{\tilde{\sigma}_2^*}[\hat{x}p^*(\hat{x})]\\&>\alpha(x_1')E_{\tilde{\sigma}_2^*}[\hat{x}p^*(\hat{x})]=x_1'
        \end{align*}
        which is a contradiction. So I conclude that $x_1\in\supp\tilde{\sigma}_2^*$. Since $\bar{x}\in\supp\tilde{\sigma}_2^*$, by proposer optimality in $t=2$, (observe that, by claim 3, the proposer's expected payoff in $t=2$ is $(1-\bar{x})V$)
        \begin{align*}
            (1-x_1)p^*(x_1)V=(1-\bar{x})V\iff p^*(x_1)=\frac{1-\bar{x}}{1-x_1}.
        \end{align*}
        Now, by proposer optimality in $t=1$,
        \begin{align*}
            (1-x_1)p^*(x_1)+(1-\bar{x})(1-p^*(x_1))&=(1-x_1')p^*(x_1')+(1-\bar{x})(1-p^*(x_1'))\\(1-\bar{x})+(1-\bar{x})\frac{\bar{x}}{1-x_1}&=(1-\bar{x})+p^*(x_1')(\bar{x}-x_1')\\p^*(x_1')&=\frac{(1-\bar{x})\bar{x}}{(1-x_1)(\bar{x}-x_1')}.
        \end{align*}
        But then observe that $p^*(x_1')>\frac{1-\bar{x}}{1-x_1'}$ since
        \begin{align*}
            \frac{(1-\bar{x})\bar{x}}{(1-x_1)(\bar{x}-x_1')}&>\frac{1-\bar{x}}{1-x_1'}\\\iff \bar{x}(1-x_1')&>(1-x_1)(\bar{x}-x_1')\\\iff \bar{x}-\bar{x}x_1'&>\bar{x}-x_1'-\bar{x}x_1+x_1x_1'\\\iff x_1'(1-\bar{x})&>x_1(x_1'-\bar{x})
        \end{align*}
        which holds since the left hand side is strictly positive and the right hand side is strictly negative, since $x_1<x_1'<\bar{x}$. Since $p^*(x_1')>\frac{1-\bar{x}}{1-x_1'}$, the proposer has a profitable deviation in $t=2$ by offering $x_1'$ because
        \begin{align*}
            (1-x_1')p^*(x_1')>(1-x_1')\frac{1-\bar{x}}{1-x_1'}=1-\bar{x}.
        \end{align*}
    \end{proof}

        For (1), let $x_1\in\supp\sigma_1^*$. Suppose for sake of contradiction that $x_1>E_{\sigma_2^*[x_1]}[\hat{x}p^*(\hat{x})]$. Then $p^*(x_1)=1$ by respondent optimality, contradicting Claim 1. 
        
        For (2), observe that $\bar{x}\not\in\supp\sigma_1^*$ since $p^*(x_1)<1$ for all $x_1\in\supp\sigma_1^*$ by Claim 1. Fix $x_1\in\supp\sigma_1^*$. Then if $x_2\in\supp\sigma_2^*$ and $x_2\neq\bar{x}$, then $x_2<\bar{x}$. Therefore, $p^*(x_2)<1$ and $\alpha(x_2)>0$. Therefore, $x_2\in\supp\sigma_1^*$. It follows that $\supp\sigma_2^*[x_1]\subseteq\supp\sigma_1^*\cup\{\bar{x}\}$.

        For (3), by Claim 5 $\exists x_L\in[0,1)$ such that $\supp\sigma_1^*=\{x_L\}$. Denote the on-path offer distribution in $t=2$ by $\tilde{\sigma}_2^*$. By (2), $\supp\tilde{\sigma}_2^*\subseteq\supp\sigma_1^*\cup\{\bar{x}\}=\{x_L,\bar{x}\}$. Since $\bar{x}\not\in\supp\sigma_1^*$ by claim 1, then $\bar{x}>x_L$. By claim 3, $\bar{x}\in\supp\tilde{\sigma}_2^*$. Since $p^*(x_L)>0$, then $\alpha(x_L)\neq1$ and so $x_L\in\supp\tilde{\sigma}_2^*$. Letting $x_H=\bar{x}$, we have that $\supp\tilde{\sigma}_2^*=\{x_L,x_H\}$.
\end{proof}

\end{document}